\documentclass[a4paper,12pt]{article}
\usepackage{graphicx}
\usepackage{multirow}
\usepackage{bbm}
\usepackage{float}
\usepackage{footnote}
\usepackage{amssymb}
\usepackage{pifont}
\usepackage{ifpdf}
\ifpdf
\else
\usepackage{pstcol,pst-fill,pstricks}
\fi
\usepackage{natbib}
\usepackage{mathpazo}
\usepackage{import}
\usepackage{dsfont}
\usepackage{enumerate}
\usepackage{a4}
\usepackage{lscape}
\usepackage{epsfig}
\usepackage[latin1]{inputenc}
\usepackage[OT1]{fontenc}
\usepackage[T1]{fontenc}
\usepackage{color}

\usepackage[hyperfootnotes=false]{hyperref}
 \hypersetup{colorlinks=true,
 linkcolor=blue,
 citecolor=blue,
 filecolor=black,
 urlcolor=black,
 pdfstartview={Fit},
pdfpagemode=UseNone
 }

\usepackage{geometry}
\usepackage{mdframed}
\usepackage{tikz}
\usepackage{tkz-graph}
\makeatletter
\newcommand*\bigcdot{\mathpalette\bigcdot@{.5}}
\newcommand*\bigcdot@[2]{\mathbin{\vcenter{\hbox{\scalebox{#2}{$\m@th#1\bullet$}}}}}
\makeatother
\usepackage{authblk}
\usepackage{setspace}
 \newcommand{\be}{\begin{equation}}
\newcommand{\ee}{\end{equation}}
\newcommand{\E}{\mathbb{E}}
\usepackage{mathtools}
\usepackage{amsmath}
\usepackage{amsthm}
\usepackage{thmtools}
\usepackage{calrsfs}
\DeclareMathAlphabet{\pazocal}{OMS}{zplm}{m}{n}
\allowdisplaybreaks

\declaretheoremstyle[
  headfont=\normalfont\scshape,
  numbered=unless unique,
  bodyfont=\normalfont,
  spaceabove=1em,
  prefoothook=\newline\rule{\linewidth}{1pt},
  spacebelow=1em,
]{exmpstyle}

\declaretheorem[
  style=exmpstyle,
  title=\textbf{Definition},
  refname={example,examples},
  Refname={Example,Examples}
]{ddef}

\newcommand{\p}{\partial}

\newtheorem{Ass}{Assumption}
\newtheorem{prop}{Proposition}

\theoremstyle{definition}

\title{Simple Analytics of the Government Investment Multiplier\thanks{We thank He Nie and Zheng Zhongxi for excellent research assistance. We would also like to thank Thiago Teixeira Ferreira, Bruce Preston, Iacopo Varotto as well as seminar participants at the 2022 ICMAIF and EEA-ESEM conferences, the 2023 ZEW workshop, ICMAIF and Midwest Macro conferences as well as seminar particiants at the Monetary Authority of Singapore. A special thanks goes to Hafedh Bouakez for extensive comments and suggestions on an early version of the current paper. Finally, we also thank anonymous referees for comments on a previous version of this paper.}}
\author[1]{Chunbing Cai}
\author[2]{Jordan Roulleau-Pasdeloup}
\affil[1,2]{Department of Economics, National University of Singapore}
\date{\today}

\begin{document}
\begin{titlepage}
\maketitle
\thispagestyle{empty}

\abstract{What are the effects of investing in public infrastructure? We answer this question with a New Keynesian model. We recast the model as a Markov chain and develop a general solution method that nests existing ones inside/outside the zero lower bound as special cases. Our framework delivers a simple expression for the contribution of public infrastructure. We show that it provides a unified framework to study the effects of public investment in three scenarios: $(i)$ normal times $(ii)$ short-lived liquidity trap $(iii)$ long-lived liquidity trap. We find that calibrations commonly used lead to multipliers that diverge with the duration of the trap.}\\[.5cm]
\noindent{\bfseries JEL Codes: E3;C62;C68;D84} \\
\noindent{\bfseries Keywords: Zero Lower Bound; Public Investment; Markov chain} \\
\end{titlepage}
\setstretch{1.5}
\setlength{\parskip}{1em}

\section{Introduction}

The Global Financial Crisis of 2008 and the ensuing recession that followed sparked a renewed interest in the macroeconomic effects of fiscal policy. Indeed, fiscal policy has been used across countries when Central Banks started to run out of ammunition with interest rates close to zero. An interesting feature of these fiscal packages is that they usually feature a different composition in terms of expenditures compared to normal times. For example, the American Recovery and Reinvestment Act of 2009 allocated 40\% of non-transfer spending to government investment in public infrastructure\textemdash compared to 23\% in normal times.\footnote{See \cite{Bouakez2017} for a discussion. \cite{Bachmann2012} show that the share devoted to government investment usually rises in recessions and \cite{Bouakez2020optimal} show that this is in line with the optimal fiscal policy that a Ramsey planner would choose.} More recently, public investment packages such as the 2021 Infrastructure Investment and Jobs Act in the U.S have been passed. 

As a result, there is now a growing literature studying the macroeconomic effects of government investment. In a search for the mechanisms underpinning these empirical results, this literature has used DSGE models to incorporate government investment in public infrastructure\textemdash see \cite{Ramey2020macroeconomic} for a recent survey. Given that public investment is usually modeled as a flow that feeds into a stock of public capital/infrastructure, these papers using DSGE models have to rely on simulation results. Those that do not (see \cite{Linnemann2006} or \cite{Albertini2014} for example) have to make strong assumptions: the stock of public capital depreciates fully each period and becomes a flow by construction. These assumptions are far from innocuous for the results. The consequence is that little is known about the \emph{general} properties of government investment in DSGE models.
 
To overcome this difficulty, we develop a new analytical framework to study the effects of government investment in public infrastructure. We build on the literature that has followed \cite{Eggertsson2003} and has used simple Markov chains to represent the shocks. We show that all variables of the model can be recast as Markov chains and we use that to get simple analytical results. In particular, we provide a series of equivalence results to show that the Markov chains indeed compute a genuine equilibrium of the model. Importantly and in contrast with the literature that has followed \cite{Eggertsson2003}, our method allows for an arbitrary duration of the effective lower bound (henceforth ELB). In turn, the duration depends on both the shock persistence and magnitude. More precisely, we show how to compute the duration of the ELB in closed form for a given shock specification. 

Our method also builds on and nests \cite{Guerrieri2014} in that we assume perfect foresight \textit{before exiting} the ELB. As a result, our setup can be viewed as a unified framework to think about the cases where exit from the ELB is deterministic (\cite{Guerrieri2014}) or stochastic (\cite{Eggertsson2003}). In this paper, we focus almost exclusively on the case where exit from the ELB is stochastic instead of deterministic for the sake of space and analytical tractability. As a result, our framework nests the findings of \cite{Eggertsson2010}, \cite{Christiano2011} as well as \cite{Woodford2011} (from which the title of this paper is inspired) as a special case when government investment is not productive. This property does not hold for \cite{Guerrieri2014}. As a result, the method we propose is a genuine generalization of \cite{Eggertsson2010}, keeping the solution method fixed. Studying public investment using \cite{Guerrieri2014}'s algorithm then entails both a change in model as well as solution method compared with existing papers and we are able to sidestep this issue to get a better comparison. 

Consequently, we are able to decompose the effects of public investment in two distinct components: one that is due to public consumption and another that is due to public capital. These turn out to be quite different. In our framework that allows for an arbitrary duration of the ELB and stochastic exit, we show that the wasteful spending component is \textit{independent} of the time spent at the ELB: regardless of the assumed duration of the ELB, we find the same expressions as in \cite{Eggertsson2010} where the ELB effectively lasts for one period \textit{in expectations}. That is however not true for the second component. Using our framework, we derive a simple matrix expression that encodes the short run multiplier effects due to public capital in various situations:
\begin{align}
\text{Short run effect} = \underset{\text{gen. eq. - short run}}{\left(I-p\mathbf{A}_s\right)^{-1}}\quad\underset{\text{expectations}}{\left(\mathbf{A}_s\right)}\quad\underset{\text{gen. eq - medium run}}{\left(I-q\mathbf{A}_m\right)^{-1}}\quad\underset{\text{static}}{\left(B\right)},    
\label{eq:main_equation}
\end{align}
where $p$ is the persistence of the government investment shock and $q$ is one minus the depreciation rate of public capital. The vector $B$ encodes the static effects of public capital on consumption and inflation: it directly affects inflation
through real marginal costs in the Phillips curve. Matrix $\mathbf{A}_s$ governs how agents react to expected consumption and inflation in the short run, hence the subscript $s$. Matrix $\mathbf{A}_m$ does the same for the medium run, hence the subscript $m$. Intuitively, public capital has a static effect given by $B$ in the medium run. The general equilibrium effect, taking rational expectations into account, is given by $\left(I-q\mathbf{A}_m\right)^{-1}B$, where $I$ is the identity matrix. This effect is expected from the perspective of the short run, which gives a static effect of $\mathbf{A}_s\left(I-q\mathbf{A}_m\right)^{-1}B$ in the short run. The general equilibrium effect in the short run with endogenous expectations is then given by equation \eqref{eq:main_equation}. 

Depending on the value taken by the short/medium-run matrices, we can analyze different situations. If the ELB is never binding, then $\mathbf{A}_s=\mathbf{A}_m=\mathbf{A}$. If the ELB is binding in the short run, we study two situations. The first is a situation where there is an arbitrarily long ELB episode. In this case $\mathbf{A}_s=\mathbf{A}_m=\mathbf{A}^*$, where the last matrix encodes the fact that the Taylor rule is passive so that current inflation does not appear in the Euler equation anymore. Finally, we also consider a situation where the (expected) duration of the ELB is short: just one period. We show that here, we have an intermediate case: $\mathbf{A}_s=\mathbf{A}_*$ but now $\mathbf{A}_m=\mathbf{A}$. In the main text, we describe in detail the assumptions needed for these expressions to be valid. These boil down to simple conditions regarding the eigenvalues of matrices $\mathbf{A}$ and $\mathbf{A}^*$.

To illustrate our method, we use a simple New Keynesian model driven by a demand shock. In this context, we show that even though it has positive aggregate supply effects in the medium run, public investment is best thought of as an aggregate demand policy in the short run. If the ELB is relatively short-lived ($\mathbf{A}_m=\mathbf{A}$), then  we show how the expected persistence of the recession dynamics becomes crucial: we derive a threshold value for the persistence and explain how it depends on structural parameters. If the demand shock is both small and transitory so that expected persistence of recession dynamics is relatively short, then public investment has a larger multiplier compared to public consumption. If the shock is small so that we still have $\mathbf{A}_m=\mathbf{A}$ but more persistent so that the expected persistence of recession dynamics becomes relatively long instead; then the short run aggregate demand effects become contractionary while the medium run aggregate supply effects are still expansionary. Given these medium run effects, the impact multiplier may be less relevant compared to the cumulative Present Discount Value (PDV). 

Our framework also allows us to compute the PDV multiplier by hand. In the case of persistent recession dynamics, we derive a condition for the medium run aggregate supply effects to dominate and show that it most likely holds for usual calibrations. Finally, our framework also allows us to study what happens for a demand shock that is not too persistent but large in magnitude so that $\mathbf{A}_s=\mathbf{A}_m=\mathbf{A}^*$. We show that both aggregate supply and demand effects now turn contractionary: in this case, the government investment multiplier is clearly lower than that for public consumption.

In between these specific situations, we show that the effects of public capital depend non-trivially on the duration of the liquidity trap. Our framework delivers a closed form expression for this case that can be of independent interest. We use this expression to discuss how the government investment multiplier depends on the duration of the trap. In particular, we show that the nature of exit (deterministic versus stochastic) becomes both relevant and crucial. In addition, we show that calibrations most commonly used in the literature imply a multiplier that diverges as a function of the duration of the trap. 

To summarize, we find that public investment at the lower bound generally has relatively large effects, especially in present discount value terms. In addition, our new solution method highlights that the public investment multiplier may depend non-trivially on the duration of the trap. When that is the case, the choice of solution method (stochastic versus deterministic) becomes important and we would advise policymakers to compute multipliers using both. That way one can guarantee that the results are not reliant on one specific solution method. 

\noindent\textbf{Related Literature}\textemdash. 
This paper is related to the empirical literature that has seeked to estimate the effects of public investment/infrastructure following the seminal contribution of \cite{Aschauer1989}. Using ARRA highway grants as instruments, both \cite{Leduc2013} and \cite{Dupor2017} find limited short run effects across states, while \cite{Leduc2013} report large medium run effects. Using Mafia infiltration as an instrument in Italy, \cite{Acconcia2014} report a short run multiplier between 1.5 and 1.9. Looking at vehicle-intensive industries, \cite{Fernald1999} finds a positive effect of highways on firms' productivity. Finally, \cite{Leff2019interstate} finds a multiplier of 1.5 at an horizon of 15 years for the Interstate Highway System. See also \cite{Klein2023composition} for a recent example using a Bayesian VAR as well as well as \cite{Tervala2022building} for an analysis using Australia's recent school infrastructure reform. 

As a result, there is by now a growing literature that incorporates public investment in both Real Business Cycle and New Keynesian models following the seminal contribution of \cite{Baxter1993}. In this framework, public investment feeds into a stock of public capital that accumulates in a similar manner to private capital.\footnote{See \cite{Pappa2009}, \cite{Ganelli2010public,Ganelli2020welfare}, \cite{Leeper2010}, \cite{Coenen2012a}, \cite{Leduc2013}, \cite{Drautzburg2015}, \cite{Abiad2016},   \cite{Bouakez2017,Bouakez2020optimal}, \cite{Sims2018output},\cite{pappa2019local}, \cite{Boehm2019} \cite{Leff2019interstate}, \cite{Klein2023composition} and \cite{Tervala2022hysteresis}. See also \cite{Gallen2021transportation} for an analysis of public transportation infrastructure in a Computable General Equilibrium framework.} In turn, the stock of public capital appears in the production function and acts like a productivity shock. These papers all share the same shortcoming: they have to rely on numerical simulations since public capital adds an additional lagged state variable to the model.\footnote{\cite{Bouakez2017} do manage to get an analytical solution for the multiplier effects of government consumption/investment, but it is a bit too complicated to analyze.} Some papers (see \cite{Linnemann2006}) have studied public investment analytically by assuming that the stock of public capital depreciates every period, making it effectively a flow.\footnote{\cite{Albertini2014} also assume full depreciation of public capital but study a non-linear model nonetheless to deal with the non-linearity due to the ZLB.} In a recent contribution, \cite{Peri2022} study the effects of public investment in a production network economy.  

This paper is also related to the literature that has studied analytically the effects of public consumption through the lenses of DSGE models both in normal times and at the ELB: see \cite{Christiano2011}, \cite{Woodford2011}, \cite{Werning2011} as well as \cite{Eggertsson2010}. Our framework nests all of these as special cases if we make public investment unproductive. 

We briefly lay out a New Keynesian model with public capital in Section \ref{sec:public_capital}. In Section \ref{sec:MChain_representation}, we derive a Markov chain representation and prove our main equivalence results with respect to existing methods. In Section \ref{sec:public_investment_normal_times}, we study the aggregate effects of public investment in normal times. In Section \ref{sec:public_investment_zlb}, we study the effects of public investment in short and long-lived liquidity trap scenarios. We offer some concluding remarks in Section \ref{sec:conclusion}.

\section{A New Keynesian model with public capital}
\label{sec:public_capital}

Building on \cite{Bouakez2017} we now consider a model in which government spending is not wasteful but feeds into a stock of public capital. We leave the non-linear model to the Online Appendix and focus here on the semi log-linear approximation. Formally, the log-linear aggregate production function reads $y_t    = n_t+\epsilon_g k_t,$
where $\epsilon_g\geq 0$ is the elasticity of aggregate output $y_t$ with respect to public capital $k_t$. We assume that the stock of public capital evolves according to a standard law of motion with a depreciation rate given by $\delta\in(0,1)$. Beyond that, all model ingredients follow the standard New Keynesian model with a simple Phillips curve. More specifically, letting $c_t, \pi_t, r_t$ and $g_t$ denote consumption, inflation, the nominal interest rate and government investment respectively, the model takes the following form: 
\begin{align}
	c_t	& = \mathbb{E}_tc_{t+1} - [r_t - \mathbb{E}_t\pi_{t+1} + \xi_t] \label{eq:ll_EE} \\
	\pi_t	& = \kappa  \left[\Gamma_c c_t+\Gamma_g g_t-\Gamma_k k_t\right] \label{eq:ll_NKPC} \\
	r_t	& = \max[\log(\beta);\phi_\pi\pi_t], \label{eq:ll_TR}\\
	k_t	& = (1-\delta)k_{t-1}+\tilde{\delta}g_{t-1} \label{eq:ll_K_LoM}\\
    y_t	& = s_cc_t + g_t \label{eq:ll_RC} 
\end{align}
where $\xi_t$ is a preference shock, $\tilde{\delta}=\delta/(1-s_c)$ and $s_c$ is the steady state share of consumption in GDP. The introduction of parameters $\Gamma$ allow us to move from different types of preferences. For example, with log consumption utility and convex labor disutility with an inverse Frisch elasticity of $\eta$, we get $\Gamma_c=(1+\eta s_c)$, $\Gamma_g=\eta$, $\Gamma_k=(1+\eta)\epsilon_g$, where  $\epsilon_g$ is the elasticity of aggregate output with respect to public capital for a given labor input. Using these parameters, we can also nest the specification in \cite{Bouakez2017} as well as the recent contribution of \cite{Klein2023composition}. For simplicity and for future reference, we collect all the model parameters in the vector $\theta$. Following \cite{Bilbiie2019}, we have assumed that firms face a \cite{Rotemberg1982}-style adjustment cost based on yesterday's \textit{market} price and not their own price. This is chiefly done for analytical tractability and the method described in the next section also naturally applies to the more standard model. Finally, notice that with $\epsilon_g=0$ we nest the standard model with public consumption.

\section{An Exact Markov Chain Representation}
\label{sec:MChain_representation}

The model can be recast as two different systems of stochastic difference equations depending on whether the ELB binds or not. Throughout the paper, we run experiments where the ELB binds for $T\in [0,\infty)$ periods. Whenever the ELB doesn't bind, we can rewrite the forward looking part of the model as:
\begin{align}
\label{eq:LRE_inertial_fwd}
\mathbf{A}_0Y_{t} &= \mathbf{A}_1\E_t Y_{t+1}+B_0k_t+C_{0,1}g_t+C_{0,2}\xi_t
\end{align}
for $t=T+1,\dots$ where $Y_t$ is a column vector of size $N_y=2$ and both $g_t$ and $\xi_t$ are assumed to follow $AR(1)$ processes with persistence $p\in (0,1)$. Whenever the constraint is binding, equation \eqref{eq:LRE_inertial_fwd} becomes: 
\begin{align}
\label{eq:LRE_inertial_cstrt_fwd}
\mathbf{A}^*_{0}Y_{t} &= \mathbf{A}^*_{1}\E_t Y_{t+1}+B_0^*k_t+C^*_{0,1}g_t+C^*_{0,2}\xi_t+E^*
\end{align}
for $t=1,\dots,T$ if $T\geq 1$. Given that the law of motion for $k_t$ (equation \eqref{eq:ll_K_LoM}) does not depend on $Y_t$ explicitly, there is no endogenous persistence. In our case, the law of motion for $k_t$ is the same inside and outside the ELB.

Our main goal is to develop a solution method that computes an equilibrium with an occasionally binding ELB and that also allows for stochastic exit out of said ELB. We will build on the literature pioneered by \cite{Eggertsson2003} that has fruitfully made use of Markov chains. Our setup is different in that the presence of public capital gives an exogenous variable with more inertia: it follows an $ARMA(2,1)$. As a result, we cannot rely on the results in \cite{Eggertsson2010} and we will have to define a more general set of Markov chains. We describe these in the following definition. 

\begin{ddef}[Markov chain representation]
\label{ddef:MChains_inertial}
Let us define $N_y+3=5$ Markov chains $\mathbf{C}_{t},\mathbf{\Pi}_{t},\mathbf{G}_{t},\mathbf{\Xi}_{t},\mathbf{K}_{t}$. We let uppercase bold letters indexed by time $t$ denote the Markov chains themselves. The Markov states reached by these chains are denoted with lowercase plain letters. All Markov chains will follow the same transition matrix given by:
\begin{align*}
\mathcal{P}=
\begin{bmatrix}
0 & 1 & 0 & \dots & \dots & 0  \\
\vdots & \ddots & \ddots &  \ddots &  & \vdots \\
0 & \dots  & 0 & 1 & 0 & \vdots\\
0 & \dots &  0 &  p & 1-p & 0\\
0 & \dots &  \dots &  0 & q & 1-q\\
0 & \dots  & \dots & \dots & 0 & 1
\end{bmatrix}
\end{align*}
which is a stochastic matrix of size $(L+3)\times(L+3)$. The 1's on the first $L$ off-diagonals imply that we assume perfect foresight during the first $L$ periods. The initial distribution for each Markov chain is given by $u=\left[1,0\dots,0\right]$
of size $1\times(L+3)$. This implies that all Markov chains start in the first state. The matrix of Markov states for all the Markov chains representing the forward looking variables is given by:
\[
\mathbf{Y}=
\begin{bmatrix}
c_{1} & \dots & c_{\ell} & \dots & c_{L+2} & 0\\
\pi_{1} & \dots & \pi_{\ell} & \dots & \pi_{L+2} & 0\\
\end{bmatrix}
\]
which we have to solve for. We define $Y_\ell = \left[c_\ell, \pi_\ell\right]^\top$ for convenience. The Markov states for the Markov chains $j$ representing the exogenous variables are set up so that these exactly replicate the $AR(1)$ dynamics for the preference and government spending shocks as well as the $ARMA(2,1)$ dynamics for the stock of public capital. These are described in detail in Online Appendix \ref{sec:app_MChains_inertial}. 
\end{ddef}
With this definition in hand, we can solve the model backwards. The objective is to solve for the matrix of Markov states for forward looking variables $\mathbf{Y}$. Assuming that this has been done, the main object of interest in this paper will be the impulse response after the two correlated shocks. Defining an impulse response for a Markov chain is not straightforward and therefore we give a precise definition here.\footnote{This definition is consistent with the definition of a Generalized Impulse Response Function given in \cite{Koop1996impulse}. In our case, if $\xi_1=g_1=0$ then $\E_t \mathbf{X}_{t+n}=0$ for all $n\geq 0$. As a result, our impulse response is not written as a difference of expected paths simply because the baseline path is a vector of zeros.}
\begin{ddef}[Impulse Response]
Let us denote by $\pazocal{I}_{x,t+n}(\xi_1,g_1,\theta)$ the impulse response function for variable $x$. Throughout the paper, we define
\begin{align*}
\pazocal{I}_{x,t+n}(\xi_1,g_1,\theta)\equiv \E_t \mathbf{X}_{t+n},    
\end{align*}
where $\mathbf{X}_{t}$ is the Markov chain associated with variable $x$ and the conditional expectation on the right hand side is taken with respect to the (common) initial distribution $u$ for all Markov chains.  
\end{ddef}
The results in \cite{Roulleau2023analyzing} guarantee that in the special case of $L=0$ when $\xi_1$ is such that the ELB is never binding, this impulse response function solves the linear model given by the backward equation \eqref{eq:ll_K_LoM} and the system of forward equations \eqref{eq:LRE_inertial_fwd} for a given government spending shock $g_1$.\footnote{In \cite{Roulleau2023analyzing} the focus is on models with AR(1) shocks and one endogenous state variable. The framework in the current paper considers a model with AR(1) shocks but because public capital is not tied to decision variables, it is effectively exogenous. As a result, the current framework is a special case of the one studied in \cite{Roulleau2023analyzing}.} However, these results do not apply when the economy is at the ELB and is expected to be there for more than 1 period. In the following Proposition, we show that the Markov chains described in Definition \ref{ddef:MChains_inertial} do solve the model for an arbitrary duration of the ELB.  
\begin{prop}
\label{prop:match_IRF}
Assume that $\xi_1$ is set such that the ELB binds for $T\geq 1$ periods in expectations. It then follows that $\pazocal{I}_{x,t+n}(\xi_1,g_1,\theta)$ for $x\in\left\{c,\pi,\xi\,g,k\right\}$ solves the system of equations given by the backward equation \eqref{eq:ll_K_LoM} and the forward equation \eqref{eq:LRE_inertial_cstrt_fwd} for all $n\geq 0$.
\end{prop}
\begin{proof}
See Online Appendix \ref{sec:proof_match_IRF}.
\end{proof}
In a nutshell, Proposition \ref{prop:match_IRF} guarantees that the Markov chains with the same transition matrix all solve\footnote{To be more precise, it solves for the Minimum State Variable solution of the underlying model.} the system of backward and forward equations \eqref{eq:ll_K_LoM}, \eqref{eq:LRE_inertial_fwd} and \eqref{eq:LRE_inertial_cstrt_fwd}. Note that we haven't specified the nature of the exit yet so that Proposition \ref{prop:match_IRF} applies both to stochastic and deterministic exits from the ELB. The following definition characterizes the main object of interest in this paper: the (marginal) government investment multiplier at the ELB:
\begin{ddef}
The impact multiplier effect for variable $x$ is then defined as:
\begin{align*}
\pazocal{M}_x(L,\theta) \equiv \lim_{g_1\to 0}\frac{\pazocal{I}_{x,t}(\xi_1,g_1,\theta)-\pazocal{I}_{x,t}(\xi_1,0,\theta)}{g_1},  
\end{align*}    
which can also be interpreted as $\partial \pazocal{I}_{x,t}(\xi_1,g_1,\theta)/\partial g_1$.
\end{ddef}

Note that this case nests the multiplier in normal times as a special case: in that case, the second impulse response on the numerator is given by zero: the multiplier obtains simply by dividing the impact of the impulse response by $g_1$. We now describe how we will solve for the Markov states to guarantee that the impulse responses actually solve the model in and out of the ELB. To solve for $\mathbf{Y}$, we work backwards and solve for $Y_{L+2}=\left[c_{L+2}, \pi_{L+2}\right]^\top$ first, when the ELB is not binding anymore. This will be the case in all the situations that we will study in this paper. For our simple New Keynesian model $\mathbf{A}_0$ is non-singular and thus we can define $\mathbf{A}\equiv \mathbf{A}_0^{-1}\mathbf{A}_1$ as well as $B\equiv \mathbf{A}_0^{-1}B_0$. Given this, we can solve for $Y_{L+2}$ as follows:
\begin{align}
Y_{L+2} &= q\mathbf{A}Y_{L+2}+Bk_{L+2}= \left(I-q\mathbf{A}\right)^{-1}Bk_{L+2},
\label{eq:YLp2_kLp2}
\end{align}
where we have used the fact that $g_{L+2}=\xi_{L+2}=0$. Given a government investment shock $g_1$, the whole path of public capital is determined. As a result, $k_{L+2}$ is effectively given and one can immediately solve for $Y_{L+2}$. What happens before state $L+2$ crucially depends on the size, sign and persistence of the shocks as well as the nature of exit. To get a set of baseline results first, we will assume that $\xi_1=0$ so that the ELB is never binding. Effectively, this amounts to study the effects of public investment in normal times, away from the ELB.

\section{Public Investment in Normal Times}
\label{sec:public_investment_normal_times}

When $\xi_1=0$ so that there is no preference shock, we can set $L=0$ without loss of generality: the ELB binds for $T=0$ periods. In that case, we only have to solve for two vectors of forward-looking variables: $Y_1$ and $Y_2$. The first one will encode short run dynamics when only government investment expenses appear in the Phillips curve, while the second one will encode medium run dynamics when only the stock of public capital appears in the Phillips curve. In that case, the transition matrix boils down to 
\begin{align*}
\mathcal{P} = 
\begin{bmatrix}
p & 1-p & 0\\
0 & q & 1-q\\
0 & 0 & 1
\end{bmatrix}
\end{align*}
as in \cite{Roulleau2023analyzing}, except that $q$ is given here and does not need to be solved for. We will solve the model backward and start by expressing consumption and inflation in $Y_2$ as a function of the stock of capital. Given the one period time to build delay, we have $k_1=0$ and thus we only have one Markov state for the stock of capital, which we will simply call $k_2=k$ for convenience. Setting $L=0$ and using equation \eqref{eq:YLp2_kLp2}, we get:
\begin{align}
\label{eq:YL2_k}
Y_{2} &=   \left(I-q\mathbf{A}\right)^{-1}Bk.  
\end{align}
In this 2 by 2 linear system of equations, $Bk$ represents the static effects of public capital on the aggregate supply side of this economy in the medium run. Multiplying by $\left(I-q\mathbf{A}\right)^{-1}$ amounts to compute the general equilibrium effect taking endogenous expectations into account. Indeed, conditional on being in the second state we can write $\E_{2}\mathbf{C}_{t+1} = qc_2$, where $\E_{2}$ denotes expectations conditional on being in state 2. We can proceed in the same manner for inflation expectations.\footnote{Note that since monetary policy is active, both the eigenvalues of $q\mathbf{A}$ are within the unit circle and there is no singularity as in \cite{Carlstrom2015inflation}. The fact that the \cite{Blanchard1980solution} condition is respected implies that the eigenvalues of $q\mathbf{A}$ are within the unit circle for $q\in(0,1)$.} Computing the matrix multiplication in equation \eqref{eq:YL2_k} gives us medium run consumption and inflation as a function of $k$:
\begin{align*}
c_2 &= \frac{\kappa(\phi_\pi-q)\Gamma_k}{\det\left(I-q\mathbf{A}\right)}k \equiv \Theta_{c,k}\cdot\epsilon_g\cdot k\quad\&\quad\pi_2 = \frac{-\kappa(1-q)\Gamma_k}{\det\left(I-q\mathbf{A}\right)}k \equiv \Theta_{\pi,k}\cdot\epsilon_g\cdot k,    
\end{align*}
where we have used the fact that $\Gamma_k$ is a linear function of $\epsilon_g$. Given that both eigenvalues of $q\mathbf{A}$ are in the $(0,1)$ interval, it can be shown that $\det\left(I-q\mathbf{A}\right)>0$. As a result, the common denominator to both of these fractions is necessarily positive. From these two equations, one can directly see that the newly defined coefficients on the right hand side are such that: (1) $\Theta_{c,k}>0$ (2) $\Theta_{\pi,k}<0$ and (3) $\Theta_{c,k} = -\frac{\phi_\pi-q}{1-q}\Theta_{\pi,k}$, which implies that $|\Theta_{c,k}|>|\Theta_{\pi,k}|$.

In words, this means that the effects of public capital in the medium run are unambiguously positive on consumption and unambiguously negative on inflation: an increase in public capital is decidedly an aggregate supply-type policy in the medium run. Fact number (3) will be useful shortly when we will compute the short term effects of public capital that work through rational expectations. For now, note that it must be the case that the effect on consumption in the medium run is larger in magnitude than the one on inflation as long as the Taylor principle holds. Note also that the effects of $k$ can be mapped out to the initial government investment shock as $k=\tilde{\delta}g/(1-p)$, where $g\equiv g_1$ is the initial government spending shock. 

Moving on to the short run, since $k_1=0$ only the government spending shock $g$ will appear. With this in mind and using the Taylor rule as well as the Markov structure to compute expectations, the relevant Markov states solve the following system of two equations:
\begin{align}
\label{eq:c_1_Ntimes}
c_1 &= \overbracket{pc_1 + (1-p)c_2}^{\E_1\mathbf{C}_{t+1}} -(\phi_\pi \pi_1 \overbracket{- p\pi_1 - (1-p)\pi_2}^{-\E_1\mathbf{\Pi}_{t+1}})\\
\pi_1 &= \kappa\left[\Gamma_c c_1 + \Gamma_g g\right].
\label{eq:pi_1_Ntimes}
\end{align}
To gain some insight on how the contribution of public capital relates to the one made by public consumption, it is useful to rewrite these two equations in matrix form as follows:
\begin{align*}
Y_1 &= p\mathbf{A}Y_1+(1-p)\mathbf{A}Y_2+C_{1}g\\
&= \left[(I-p\mathbf{A})^{-1}C_1 + (I-p\mathbf{A})^{-1}\mathbf{A}(I-q\mathbf{A})^{-1}B\Tilde{\delta}\right]g,
\end{align*}
where the vector $C_1$ denotes the first column of the matrix $\mathbf{A_0}^{-1}C_{0,1}$\textemdash see equation \eqref{eq:LRE_inertial_fwd}. The first term in brackets on the right hand side is the solution that obtains in the standard New Keynesian model with public consumption. The second part governs the contribution of public capital. This equation shows that these two effects appear linearly and thus can be analysed in isolation. Focusing on the contribution of public capital, one can see that we recover the matrix expression described in the introduction.

Given our assumption of a static Phillips curve, government investment has less of an aggregate supply flavor to it in the short run. Indeed, the expected decrease in inflation in the medium run does not feed into lower inflation in the short run through the firms' decisions. The aggregate supply effects of government investment kick in in the medium run and  influence short run outcomes exclusively through the Euler equation. There, as can be seen from equation \eqref{eq:c_1_Ntimes}, both expectations play conflicting roles: higher medium run consumption incentivizes higher consumption in the short run. Lower medium run inflation however brings about a larger expected real interest rate and thus disincentivizes consumption in the short run. From fact 3 above however, we know that the consumption effect dominates so that public capital makes a positive contribution on short run consumption everything else equal.\footnote{\cite{Klein2023composition} explicitly assume this property in their model with a forward-looking Phillips curve. See their assumption \textbf{A2}.} 

Beyond the specific contribution of public capital, we want to know whether government investment can be a useful tool to stimulate an economy. To do that, we need to put the pieces back together and add up the direct effects of $g$ with its indirect effects through $k$. This is done in the following proposition.

\begin{prop}
Assume that $\phi_\pi>1$. Then, it follows that:
\begin{enumerate}
    \item 
There exists a threshold $\epsilon_g^I$ such that 
$\frac{\partial c_1}{\partial g}>0\ $ if  $\ \epsilon_g > \epsilon_g^I>0$    
\item  The threshold value for $\epsilon_g$ is given by:
\begin{align*}
\epsilon_g^I = \frac{\phi_\pi-p}{\phi_\pi-1}\frac{\Gamma_g}{1+\Gamma_g}\frac{1-q+\kappa\Gamma_c(\phi_\pi-q)}{\tilde{\delta}}.    
\end{align*}
\item The government investment multiplier effect on inflation $\pi_1$ is strictly positive.
\end{enumerate}
\label{prop:pi_1_NTimes}
\end{prop}
\begin{proof}
See Online Appendix \ref{sec:proof_pi1_NTimes}.
\end{proof}
Part 1 of Proposition \ref{prop:pi_1_NTimes} says that a government investment shock can crowd in private consumption only if it is assumed to be productive enough. As a corollary, from the resource constraint equation \eqref{eq:ll_RC} it is evident that whenever this condition is cleared the government investment multiplier effect on output $\partial y_1/\partial g$ is strictly larger than 1. This results also nests the well known result that public consumption crowds out private consumption as this case obtains with $\epsilon_g=0$ which violates the condition in Part 1 of Proposition \ref{prop:pi_1_NTimes}.

Part 2 of Proposition \ref{prop:pi_1_NTimes} delineates how the threshold value for the productivity of public capital depends on the structural parameters of the model. One parameter that stands out here is $\Gamma_g$, the elasticity of marginal costs with respect to government investment: the lower this parameter, the lower the threshold for $\epsilon_g$. Indeed, a low value for this parameter requires a large Frisch elasticity of labor supply. The low value of $\Gamma_g$ guarantees that a government investment shock has little effect on marginal costs and thus inflation in the short run. This results in a very steep Phillips curve in a $(\pi_1,c_1)$ plane. In addition, a positive government spending shock shifts this line to the right: given consumption, inflation has to increase. Without public capital, this rise in inflation fuels an increase in the expected real interest rate through the reaction of the Central Bank: consumption necessarily declines in the short run. This is illustrated in the right panel of Figure \ref{fig:AS_AD_NTimes}. Importantly, our method ensures that the lines plotted in this figure take full account of expectations.

The dynamics change however when public capital is brought into the picture: the expected wealth effect from higher consumption in the medium run makes the Euler equation shift up. Given fixed inflation today, consumption increases in the short run. If this upward shift is large enough, consumption is crowded in in the short run. If $\Gamma_g$ is too high so that $\epsilon_g<\epsilon_g^I$, then the Phillips curve effectively shifts more to the right: the government spending shock generates too much inflation in the short run. The real rate effect in the Euler equation then gets magnified and more than compensates for the expected wealth coming from higher capital in the medium run. This expected wealth effect has been exposed  in \cite{Bouakez2017} and shown to dominate the aggregate supply effect in the medium run numerically. In our framework, we zoom in on the aggregate demand effects of government investment and show in closed form how big they need to be to crowd in consumption in the short run.

Part 3 of Proposition \ref{prop:pi_1_NTimes} guarantees that the impact effect of government investment on inflation is strictly positive regardless of the effect on short run consumption.\footnote{This result is consistent with the findings described in \cite{Morita2020empirical} as well as \cite{Klein2023composition}.} Again, this can be clearly seen in Figure \ref{fig:AS_AD_NTimes}. Both lines shift to the right so that inflation necessarily increases. Intuitively, the usual public consumption effect increases marginal costs and thus inflation through the Phillips curve. On top of that, the expected wealth effect from the Euler equation acts as an expansionary demand shock, which also pushes inflation upwards.  

\begin{figure}[htp]
\centering
\caption{Short run effects of government investment}
\includegraphics[width=0.8\textwidth]{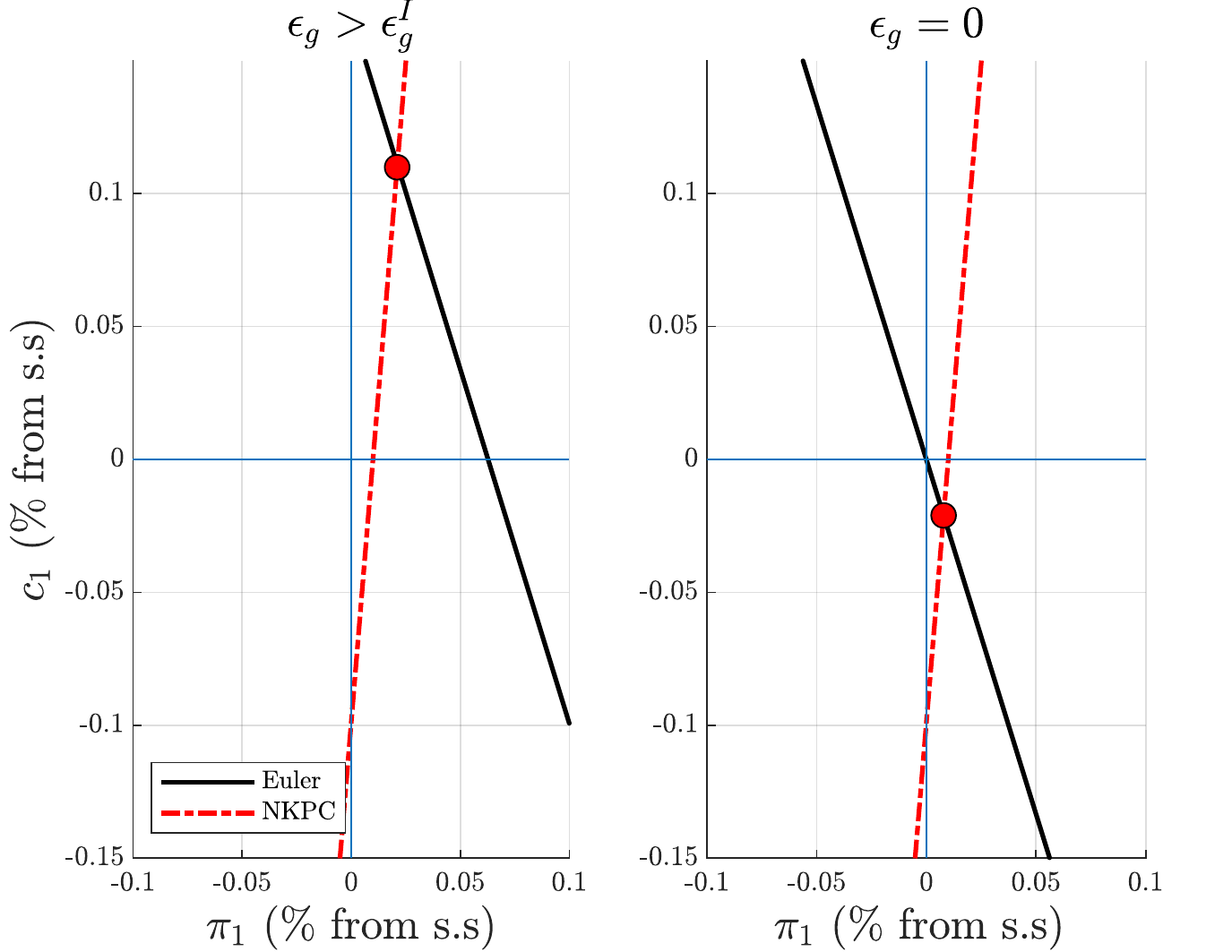}
\begin{minipage}{0.8\textwidth}
\protect\footnotesize Notes: \protect\scriptsize
We use the following calibration: $\beta=0.99, \kappa=0.1, \eta=0.01, \phi_\pi=1.5, q = 0.95$ as well as $p=0.7$.
\end{minipage}
\label{fig:AS_AD_NTimes}
\end{figure}

The tools that we have developed enable us to take a model that has usually been studied numerically and compute analytical results. For example, in a working paper version of \cite{Bouakez2017} (see \cite{Bouakezpublic}), the authors follow \cite{Linnemann2006} and solve a version of their model where public \textit{investment} (and not capital) appears in the production function. In their case, the presence of productive public spending can be found to reverse the direction of the shift in the Phillips curve. Our results show that this is not a good approximation of what is actually happening in the model with capital, as the latter manifests itself exclusively as a shift of the Euler equation. 

The point of Proposition \ref{prop:pi_1_NTimes} is not to argue that the scenario on the left panel is more likely than the one on the right. In this model, the configuration on the left requires a very specific set of parameters in the context of the standard NK model. This may however be different in a model that has amplification or dampening due to the presence of hand to mouth agents as in \cite{Bilbiie2020new}. Rather, the goal of Proposition \ref{prop:pi_1_NTimes} is to clearly delineate the conditions under which public investment can crowd in private consumption so that these can be checked in any variant of the standard New Keynesian model.

However, this is only the impact effect. As a result, it does not do full justice to the effects of public investment because the stock of public capital will stay elevated long past the first period. The cumulative present discount value (PDV) multiplier might then be a better measure of the overall impact of government spending then. Fortunately, our framework lends itself naturally to a simple computation of the cumulative PDV effects of government investment. Indeed, given our assumption that the ELB is never binding, we can still use the formulas derived in \cite{Roulleau2023analyzing}. The properties of the PDV multiplier are exposed in the following proposition. 
\begin{prop}
\label{prop:Mc_NTimes}
Let us define $\Theta$ such that $c_2+\pi_2 = \Theta \epsilon_g\Tilde{\delta}g/(1-p)$. The cumulative present discount value multiplier is then given by:
\begin{align}
\nonumber
\mathfrak{M}_c &= \frac{\partial c_1}{\partial g} + \beta\frac{1- p}{1-\beta q}\frac{\partial c_2}{\partial g}\\
&= \frac{\Theta\tilde{\delta}\epsilon_g-\kappa(\phi_\pi-p)\Gamma_g}{\det(I-p\mathbf{A})}+\frac{\beta}{1-\beta q}\Theta_{c,k}\tilde{\delta}\epsilon_g.
\label{eq:PDV_c}
\end{align}
It follows that there exists a threshold $\epsilon^M_g$ such that $\mathfrak{M}_c>0$ if $\epsilon_g>\epsilon^M_g$. We also have:
\begin{align*}
\epsilon^M_g = \frac{\epsilon_g^I}{1+\frac{\beta}{1-\beta q}\frac{\phi_\pi-q}{\phi_\pi-1}\det(I-p\mathbf{A})}<\epsilon_g^I.    
\end{align*}
\end{prop}
\begin{proof}
See Online Appendix \ref{sec:proof_Mc_NTimes}.
\end{proof}
Given that the first term on the right hand side of equation \eqref{eq:PDV_c} is $\partial c_1/\partial g_1$, the PDV multiplier adds a second term that depends positively on $\epsilon_g$. Intuitively, beyond its expected effects in the short run public capital has long lasting effects on the productive capacity of this economy. This generates a higher level of consumption that persists for a long time. It then follows that a lower value of $\epsilon_g$ is required to get a positive PDV effect. The second part of Proposition \ref{prop:Mc_NTimes} describes exactly how much lower this threshold has to be compared to the one introduced in Proposition \ref{prop:pi_1_NTimes}. Given that both $\beta$ and $q$ are usually very close to 1, the second term on the denominator of $\epsilon^M_g$ can be sizable. For most calibrations that we have tried,  $\epsilon^M_g$ is usually one order of magnitude lower compared to $\epsilon^I_g$.

The quantitative properties of this model are largely consistent with the existing studies on the topic such as \cite{Leeper2010}, \cite{Bouakez2017}, \cite{Ganelli2020welfare}, \cite{Sims2018output}, \cite{Ramey2020macroeconomic} and \cite{Klein2023composition}. Our contribution is to derive a clear threshold for consumption to be crowded in, both in terms of the impact and the present discount value multiplier. 

Ultimately however, as the results in \cite{Bachmann2012} exemplify, policymakers are usually interested in public investment as a stimulative tool in a recession, especially if the ELB binds.\footnote{Another metric of interest to policymakers is the welfare impact of public investment. See \cite{Sims2018output}, \cite{Ganelli2020welfare} and \cite{Bouakez2020optimal}.} This is where the new method that we develop in this paper will be useful. Accordingly, we now move on to describe how these tools can be used to study the effects of public investment at the ELB.

\section{Public Investment at the ELB}
\label{sec:public_investment_zlb}

To study the effects of public investment at the ELB, we cannot rely on the results in \cite{Roulleau2023analyzing} anymore. For a given value of $\xi_1$, the ELB will be binding for a number of $T\geq 1$ periods, which needs to be determined. Loosely speaking, we can make two assumptions here. If we assume a deterministic exit as in \cite{Guerrieri2014}, then we have $T=L$. If we assume a stochastic exit as in \cite{Eggertsson2010}, then we have $T=L+1$. It should be noted however that the framework in \cite{Eggertsson2010} implicitly assumes $L=0$ and does not allow for inertial exogenous states such as public capital. We will mostly focus on the case of stochastic exit and highlight how it compares to the case of deterministic exit whenever it is relevant. We make this choice because the stochastic exit will enable us to get simpler expressions for a short liquidity trap.

To compute such a solution, we will proceed with a guess and verify method: guess first that the ELB is binding for an unspecified number of periods $T$, then compute the solution according to this guess. Next, we use this solution to back out the magnitude of $\xi_1$ that is consistent with this value of $T$.

In contrast with both \cite{Erceg2014} as well as \cite{Guerrieri2014}, we are able to solve for the exact date in closed form. In fact, we only need to solve for the allocation upon exit to determine the necessary size of the shock\textemdash see the proof of Proposition \ref{prop:match_IRF} in Online Appendix \ref{sec:proof_match_IRF}. Once we have that in hand, we will work backward until the period where the shock hits. In doing so, the main difference compared to what we did in last section is that agents weigh expectations for next period with matrix $\mathbf{A}^*$ instead of $\mathbf{A}$: this encodes the fact that the Taylor rule is not active anymore at the ELB. Throughout this section, the eigenvalues of $p\mathbf{A}^*$ will be crucial in that they can potentially be larger than 1. This is described in the following proposition.
\begin{prop}
\label{prop:eigs}
Let us assume that $\kappa\Gamma_c>0$. Then $p\mathbf{A}^*$ has two eigenvalues. The first one is equal to zero and is thus stable. The other is strictly larger than 1 if, and only if
\begin{align*}
p>\frac{1}{1+\kappa\Gamma_c}\equiv \overline{p} .
\end{align*}
If that condition is met then we have $\det\left(I-p\mathbf{A}^*\right) = 1-p(1+\kappa\Gamma_c) < 0$. 
\end{prop}
\begin{proof}
See Online Appendix \ref{sec:proof_eigs}.
\end{proof}
First of all, note that $\overline{p}$ is derived in \cite{Eggertsson2010} to guarantee a unique minimum state variable solution at the ELB.\footnote{See also \cite{Bilbiie2019} for a derivation in the context of the static Phillips curve that we use here.} \cite{Eggertsson2010} uses a standard New Keynesian model where expected inflation appears in the Phillips curve, but the derivations underlying Proposition \ref{prop:eigs} carry through in this case.

If both shocks are such that $p<\overline{p}$, then $\det\left(I-p\mathbf{A}^*\right)>0$ and the Euler equation and the Phillips curve only cross once. In \cite{Mertens2014}, the fact that $p>\overline{p}$ guarantees that $\det\left(I-p\mathbf{A}^*\right)<0$ and these cross twice and sunspots can happen at the ELB. Proposition \ref{prop:eigs} then provides a connection between the methods just described and these used in \cite{Guerrieri2014}. Indeed, the method in \cite{Guerrieri2014} requires one to compute $\left(p\mathbf{A}^*\right)^L$ where $L$ is the number of periods at the ELB under deterministic exit. As a result, if one of these eigenvalues is larger than 1 then this matrix power will diverge as $L\to\infty$. Given that our method follows theirs in assuming perfect foresight at the ELB, then this will be very important. Note that the same threshold applies to $q\mathbf{A}^*$: if $q>\overline{p}$ then $\det\left(I-q\mathbf{A}^*\right)<0$ and one eigenvalue is unstable.

With this in mind, we will consider two types of situations. The first one will be a short duration of the liquidity trap where $L=0$. In that case, the ELB only binds on impact, but exit from the ELB is stochastic as in \cite{Eggertsson2010}. After that, we will consider a case where the preference shock is such that $L$ is large instead.

\subsection{A short-lived ELB episode}
\label{sec:short_lived}

The simplest way to study a short-lived liquidity trap is to assume both a stochastic exit as well as $L=0$. This is equivalent to what  has been used in the literature following \cite{Eggertsson2003} as well as \cite{Eggertsson2010}.\footnote{Alternatively, under a deterministic exit one has to consider at least $L=1$ which brings in extra effects where the stock of public capital starts to kick in during the one period at the liquidity trap.} The stochastic exit combined with $L=0$ as well as the one period time to build delay guarantees that the stock of public capital kicks in after the liquidity trap \textit{in expectations}. As such, we argue that it provides the best thought experiment to study how the timing of aggregate supply effects matters at the ELB. Under these assumptions, we can still use the same transition matrix as in Section \ref{sec:public_investment_normal_times}. When that is the case, equations \eqref{eq:c_1_Ntimes}-\eqref{eq:pi_1_Ntimes} now become:
\begin{align}
\label{eq:c_1_Short_lived}
c_1 &= \overbracket{pc_1 + (1-p)c_2}^{\E_1\mathbf{C}_{t+1}} +\overbracket{p\pi_1 + (1-p)\pi_2}^{\E_1\mathbf{\Pi}_{t+1}}-\log(\beta)-\xi\\
\pi_1 &= \kappa\left[\Gamma_c c_1 + \Gamma_g g\right].
\label{eq:pi_1_Short_lived}
\end{align}
As before, to gain some insight on how the contribution of public capital relates to the one made by public consumption, it is useful to rewrite these two equations in matrix form. Letting $Z_1 \equiv \left[g\ \xi\right]^\top$, we can rewrite equations \eqref{eq:c_1_Short_lived}-\eqref{eq:pi_1_Short_lived} as:
\begin{align*}
Y_1 &= p\mathbf{A}^*Y_1+(1-p)\mathbf{A}^*Y_2+\mathbf{C}^*Z_1 + E^*\\
&= \left[(I-p\mathbf{A}^*)^{-1}C_1^* + (I-p\mathbf{A}^*)^{-1}\mathbf{A}(I-q\mathbf{A})^{-1}B\Tilde{\delta}\right]g+\text{t.i.p}
\end{align*}
where the vector $E^*$ collects the constant term from the Euler equation and t.i.p denotes terms independent of policy which include the aforementioned constant as well as terms involving the demand shock $\xi$. This expression is very similar to the one derived in the last section, except that now we have matrix $\mathbf{A}^*$ appearing on the right hand side. Notice that regarding the medium term dynamics, it is still matrix $\mathbf{A}$ that appears given our assumption that the ELB only binds for one period. As before, the contributions of public consumption and public capital appear additively and can thus be studied separately. Notice also that nothing changes regarding how $c_2$ and $\pi_2$ relate to $k$ from the last section. The one feature that changes now is how these two feed into short term outcomes $c_1$ and $\pi_1$ through rational expectations. We describe how both consumption and inflation react in the short run after an increase in government investment in the following proposition:
\begin{prop}
\label{prop:c_1pi_1_short_lived}
We define $\Theta$ such that $c_2+\pi_2 = \Theta \epsilon_g\Tilde{\delta}g/(1-p)$. Given our assumption that the ELB does not bind in the medium run, then $\Theta\geq 0$. Then the short run effect of government investment on consumption and inflation are given by:
\begin{align*}
\frac{\partial c_1}{\partial g} &= \frac{\Theta\tilde{\delta}\epsilon_g+ p\kappa\Gamma_g}{\det (I-p\mathbf{A}^*) } \quad\&\quad
\frac{\partial \pi_1}{\partial g} = \kappa\frac{\Gamma_c\Theta\tilde{\delta}\epsilon_g + (1-p)\Gamma_g}{\det (I-p\mathbf{A}^*)},
\end{align*}
where the sign of the determinant is ambiguous. Regardless of the sign of the determinant, the effects of public capital/wasteful spending go in the same direction.  
\end{prop}
\begin{proof}
See Online Appendix \ref{sec:proof_c_1pi_1_short_lived}.
\end{proof}
The main takeaway from Proposition \ref{prop:c_1pi_1_short_lived} is that regardless of the sign of the determinant, the contribution of public capital goes in the same direction as the one from public consumption. This is in contrast with what we had before where the effects on consumption were ambiguous outside of the ELB depending on whether public capital was sufficiently productive or not.

Let us first consider the case where $p<\overline{p}$ so that both eigenvalues are stable. In that case, we have $\det (I-p\mathbf{A}^*)>0$. Given the restriction on $p$, the ELB is the result of a shock that is relatively large in magnitude but not too persistent. In this context, consumption increases on impact at the ELB after an increase in government investment. What's more, given that $\Theta$ is an increasing function of $\epsilon_g$, the more productive public capital the bigger the effect on private consumption. This is the expected wealth effect described in \cite{Bouakez2017}. In their numerical experiments, they ensure that the productive effects of public capital kick in due to a time to build delay. In our framework, we achieve this by having a liquidity trap of a short duration. As in \cite{Christiano2011}, the effect of wasteful spending is also positive in that situation: both components of government investment generate an increase in aggregate demand that pushes up both consumption and inflation in this situation.

In this scenario, considering the cumulative PDV multiplier does not bring in new insights: both the effects from public consumption as well as public capital go in the same direction. As a result, the cumulative PDV
multiplier will necessarily be of the same sign and larger in magnitude. 

The advantage of our framework is that we can also encompass the case where one of the eigenvalues of $p\mathbf{A}^*$ is unstable. To the best of our knowledge, we are the first to consider the effects of government investment in this kind of situation. 
In contrast with \cite{Mertens2014} however, we do not need to invoke an additional sunspot shock for the ELB to bind: the ELB still binds because of a decrease in the preference shock. The one difference with the situation that we have just covered is that the expected return to the steady state is now \textit{slower}. More precisely, the increased persistence of the shock makes the income effects in the Euler equation more potent, which is what generates the different results that have been reported in the literature \textemdash see in particular \cite{Bilbiie2022neo} as well as \cite{Nakata2022expectations}. In our case, the fact that $\det (I-p\mathbf{A}^*)<0$ if $p$ is large enough implies that \textit{both} public consumption as well as public capital work to make consumption decrease in the short run. 

Our results put into perspective the results of \cite{Mertens2014} or more recently \cite{Bilbiie2022neo} and \cite{Nakata2022expectations}: the presence of sunspots does not matter \textit{per se}. What does matter is the kind of mechanisms that the sunspots bring to the table. In our case, we bring these mechanisms via a more persistent fundamental shock. As a result, at a conceptual level the source of the ELB is not structural. What \textit{is} structural is the expected persistence of the recession dynamics and the mechanisms that it carries with it. Taken at face value then, the model says that if both shocks are persistent enough, government investment makes both inflation and consumption decrease at the ELB.

Compared to the case where $\det(I-p\mathbf{A}^*)>0$, when $p$ is large enough so that this condition does not hold the cumulative PDV multiplier can potentially bring non-trivial dynamics into play. In that case, it is no longer a given that the cumulative effects go in the same direction as the ones stemming from public consumption. The next proposition describes under which conditions these two scenarios arise. 

\begin{prop}
\label{prop:PDV_long_trap}
Assume that $p$ is such that $\det(I-p\mathbf{A}^*)<0$. If     
\begin{align*}
\beta\frac{\phi_\pi-q}{1-\beta q}>-\frac{\phi_\pi-1}{\det(I-p\mathbf{A}^*)}>0    
\end{align*}
then there exists a threshold $\epsilon^{M,z}_g$ such that if $\epsilon_g>\epsilon^{M,z}_g$ then $\mathfrak{M}_c>0$.
Otherwise, $\mathfrak{M}_c<\partial c_1/\partial g<0$.
\end{prop}
\begin{proof}
See Online Appendix \ref{sec:proof_PDV_long_trap}.  
\end{proof}
The first condition in Proposition \ref{prop:PDV_long_trap} ensures that the expansionary medium run effects of public capital dominate the short run contractionary ones. Given the analysis in \cite{Mertens2014} as well as \cite{Bilbiie2022neo}, the expected wealth effect is contractionary if $p$ is above the threshold for stable eigenvalues: the Euler equation slopes more than the Phillips curve at the ELB and an upward shift of the Euler equation decreases consumption. When that is the case, there is a scope for the medium to long run effects of public capital to compensate for these. If the condition spelled out in Proposition \ref{prop:PDV_long_trap} doesn't hold, then the presence of public capital only contribute negatively: in this case, the cumulative PDV multiplier is strictly lower than the (negative) impact multiplier.

The natural question is then: which case is more likely in practice? Given that public capital depreciates slowly so that $q$ is close to 1, both terms on the numerators are roughly equal. For the same reason, the term $\beta/(1-\beta q)$ is likely quite large given that $\beta\simeq 1$ in most calibrations.\footnote{This depends on the particular choice of model though. If one chooses to study a model as in \cite{Michaillat2019resolving}, which uses a low value of $\beta$ to get stable dynamics at the ELB, then the conclusion that can be drawn from Proposition \ref{prop:PDV_long_trap} may be very different. 
 The goal of this paper is to make these features as transparent as possible.} As a result, unless one works with a calibration where $\det(I-p\mathbf{A}^*)$ is both negative
\textit{and} quite close to zero, then the medium run effects most likely dominate. 

As a result, if the relevant metric is the cumulative PDV multiplier, then public investment can potentially answer the challenge laid out in \cite{Bilbiie2022neo}'s conclusion: he does not find any policy that is a good tool to stimulate an economy at the ELB regardless of the sign of $\det(I-p\mathbf{A}^*)$. Our findings echo his conclusion if we constrain ourselves to the impact multiplier. If we broaden the horizon to consider the cumulative PDV effects however, we find that there is potentially a scope for public investment to be an effective stabilization tool at the ELB.

It should be kept in mind that these conclusions are valid under the maintained assumption of a short lived ELB episode. That being said, this assumption may be at odds with the experience of the Japanese economy where expectations of zero interest rates are quite entrenched. The U.S and Eurozone countries have also spent multiple consecutive years at the ELB. Accordingly, we now turn our attention to another polar case where the shock is assumed to be large and not too persistent so that the ELB binds for several periods. 
We will show that, once again, our framework is amenable to simple analytical results.

\subsection{A long-lived ELB episode}
\label{sec:long_lived}

We now go back to our backward solution of the model. As before, we start from the second to last vector of states given by $Y_{L+2} =   \left(I-q\mathbf{A}\right)^{-1}Bk_{L+2}$
where we have used $\mathbf{A}$ because the ELB is not binding in this state. Solving for $Y_{L+1}$ is again far from straightforward and merits some discussion. Once again, we can make two choices. If we assume that the ELB is not binding, then we effectively replicate the algorithm developed in \cite{Guerrieri2014}. Instead, we can assume that the ELB is still binding in state $Y_{L+1}$. As we have seen before, this nests the framework in \cite{Eggertsson2010} in the case $L=0$. Our method can be seen as a generalization of the latter framework that allows for a generic $L$. In our case, whether we assume a deterministic or a stochastic exit will turn out to be inconsequential. To guarantee that this happens, we will rely on two assumptions:
\begin{Ass}
\label{Ass:eigenvalues}
All the eigenvalues of $p\mathbf{A}^*$ as well as $q\mathbf{A}^*$ are in the unit circle.         
\end{Ass}
\begin{Ass}
\label{Ass:big_shock}
The initial shock $\xi_1$ is arbitrarily large such that the duration of the ELB is $T\to\infty$.
\end{Ass}
We will build on the results in \cite{Erceg2014labor} or \cite{Giannoni2023forward} and assume a high degree of price stickiness to guarantee that the first assumption holds. The second assumption is meant to be an approximation: if the duration of the ELB is relatively large, then we can effectively take $T\to\infty$. With this in mind, the main result of this subsection is that under these two assumptions we can derive an "\textit{as if}" result: even though the dynamics of this economy are given by a Markov chain of size $L+3$ that tends to infinity, these are identical to the ones from a suitably defined three-state Markov chain. This is described in detail in the following proposition. 

\begin{prop}
\label{prop:long-lived_trap}
Assume that assumption \ref{Ass:eigenvalues} holds. Then there exists a unique value of $\xi_1$ such that the ELB binds for exactly $T=L+1$ periods in expectations. The vector of impact states can be written as:
\begin{align*}
Y_1 = \left[(I-p\mathbf{A}^*)^{-1}C_1^* + \pazocal{Q}(L,\theta)\right]g+\text{t.i.p},    
\end{align*}    
where, as before, $\theta$ collects all the structural parameters of the model. Furthermore, regardless of the nature of exit from the liquidity trap we have:
\begin{align*}
\lim_{L\to\infty} \pazocal{Q}(L,\theta) = (I-p\mathbf{A}^*)^{-1}\mathbf{A^*}(I-q\mathbf{A}^*)^{-1}B^*\Tilde{\delta}   
\end{align*}
\end{prop}
\begin{proof}
See Online Appendix \ref{sec:proof_long-lived_trap}.
\end{proof}
Two main results emerge from Proposition \ref{prop:long-lived_trap}. The first one is that the component associated with wasteful spending is independent\footnote{This echoes the results obtained in \cite{Carlstrom2015inflation} for an exogenous interest rate peg in the standard New Keynesian model: they find that the initial value of inflation does not depend on $L$. Our result is more general in that it does not depend on the assumption of an exogenous peg and holds for an endogenous ELB episode.} of the duration $L$. In contrast, the component associated with public capital explicitly depends on $L$. As a result, one can use the expression from Proposition \ref{prop:long-lived_trap} to study how the overall multiplier varies with $L$. It should be noted that this expression (detailed in the Online Appendix) is also valid even if Assumption \ref{Ass:eigenvalues} does not hold. In that case, the multiplier will be finite for finite values of $L$ but will diverge for $L\to\infty$. We will cover this case in a later subsection.

The second result is that if one adds Assumption \ref{Ass:big_shock}, this multiplier takes on a very simple form that once again boils down to the matrix expression given in the introduction. In this case, the model is in a quasi-permanent liquidity trap. Contrary to most of the literature, this does not come from an exogenous peg but is the result of a shock. As a consequence, whether this liquidity trap is stable will crucially depend on the persistence of this shock but not only. In fact, it also depends on $q$, which roughly measures the persistence of the stock of public capital. Lastly, the nature of exit is irrelevant because the relationship between what happens upon exit and in the short run is weighed by both $(p\mathbf{A}^*)^L$ and $(q\mathbf{A}^*)^L$. Under Assumptions \ref{Ass:eigenvalues}-\ref{Ass:big_shock}, both of these matrices converge to a matrix of zeros as $L\to\infty$.

What Proposition \ref{prop:long-lived_trap} guarantees is that this model can be studied \textit{as if} it followed a three state Markov chain where the ELB also binds in the medium run. This is the reason why matrix $\mathbf{A}$ never appears here and only $\mathbf{A}^*$ does. With this in mind, we will proceed as before and solve first for the medium run dynamics and then study how they influence the short term dynamics through expectations. Ignoring terms independent of policy, we now get:
\begin{align*}
c^z_2 &= -\frac{\kappa q\Gamma_k}{\det\left(I-q\mathbf{A}^*\right)}k \equiv \Theta^z_{c,k}\cdot\epsilon_g\cdot k\quad\&\quad    
\pi^z_2 = \frac{-\kappa(1-q)\Gamma_k}{\det\left(I-q\mathbf{A}^*\right)}k \equiv \Theta^z_{\pi,k}\cdot\epsilon_g\cdot k,    
\end{align*}
where the superscript $z$ indicates that the ELB is binding in the second state of the Markov chain. Compared with what we had in Section \ref{sec:public_investment_normal_times}, now both $\Theta^z_{c,k}$ and $\Theta^z_{\pi,k}$ are \textit{negative}. As a result, their relative magnitude does not really matter anymore: public capital makes a negative contribution to both consumption and inflation in the short run. For inflation, this comes from direct aggregate supply effects. For consumption, this is because decreased inflation at the ELB decreases consumption. Note that this crucially relies on Assumption \ref{Ass:eigenvalues} which guarantees that $\det\left(I-q\mathbf{A}^*\right)>0$. Again ignoring constant terms, we can proceed as in the short-lived scenario and compute the government spending multipliers in the long-lived scenario. These are described in the following proposition.

\begin{prop}
\label{prop:c_1pi_1_long_lived}
We now define $\Theta^z$ such that $c_2+\pi_2 = \Theta^z \epsilon_g\Tilde{\delta}g/(1-p)$. Given that the ELB does bind in the medium run, we have $\Theta^z\leq 0$. Then it follows that the short run multiplier effect of government investment on consumption and inflation is given by:
\begin{align*}
\frac{\partial c_1}{\partial g} &= \frac{\Theta^z\tilde{\delta}\epsilon_g+ p\kappa\Gamma_g}{\det (I-p\mathbf{A}^*) } \quad\&\quad
\frac{\partial \pi_1}{\partial g} = \kappa\frac{\Gamma_c\Theta^z\tilde{\delta}\epsilon_g + (1-p)\Gamma_g}{\det (I-p\mathbf{A}^*)},
\end{align*}
where the determinant is strictly positive given Assumption \ref{Ass:eigenvalues}. The effects of public capital go in the opposite direction compared to the ones from public consumption.  It follows that there exist a threshold $\epsilon_{g,c}^z$ (resp. $\epsilon_{g,\pi}^z$) such that if $\epsilon_g>\epsilon_{g,c}^z$ (resp. $\epsilon_g>\epsilon_{g,\pi}^z$) then $\partial c_1/\partial g <0$ (resp. $\partial \pi_1/\partial g <0$). Furthermore, we have
\begin{align*}
\epsilon_{g,\pi}^z= \frac{1-p}{p}\frac{1}{\kappa\Gamma_c} \epsilon_{g,c}^z.   
\end{align*}
\end{prop}
\begin{proof}
See Online Appendix \ref{sec:proof_c_1pi_1_long_lived}.
\end{proof}
What Proposition \ref{prop:c_1pi_1_long_lived} clarifies is that public capital now makes a negative contribution to both consumption and inflation in the short run. Is this contribution big enough to compensate for the positive contribution of wasteful spending on both of these variables? According to Proposition \ref{prop:c_1pi_1_long_lived}, that depends on $\epsilon_g$ as well as the other structural parameters of the model. In this model, the shift in the Euler equation from public consumption is solely due to the effect on expected inflation. As a result, it is scaled by $p$ since $\E_1\mathbf{\Pi}_{t+1}=p\pi_1 + (1-p)\pi_2$. In the Phillips curve, the shift does not depend on $p$ given the contemporaneous curve that we have assumed. The larger $p$ however, the steeper the Euler equation and any shift in the Phillips curve will have less impact on inflation. As a result, the impact of wasteful spending on inflation scales with $1-p$. In addition, $\kappa$ and $\Gamma_c$ magnify the impact of public capital on inflation through their effect on consumption in the short run. 

For the sake of the argument, let us assume that the parameters are such that $\epsilon_{g,\pi}^z>\epsilon_{g,c}^z$. In this context, we consider the following experiment: we start from the wasteful spending benchmark with $\epsilon_g=0$ and then we increase $\epsilon_g$. Initially, both consumption and inflation are crowded in. As $\epsilon_g$ crosses $\epsilon_{g,c}^z$, inflation is crowded in but consumption is crowded out. Finally, when $\epsilon_g $ passes $\epsilon_{g,\pi}^z$, both inflation and consumption decrease after an increase in government investment. For completeness, we now provide some numerical results on what happens for intermediate values of $L$.

\subsection{Results for $L\in \left[1,\infty\right)$}

So far, we have studied the simple cases of $(i)$ $L=0$ combined with stochastic exit as well as $(ii)$ $L\to\infty$ with stable eigenvalues where the nature of exit is irrelevant. In this subsection, we want to describe what happens outside of these two special cases. Given the complexity of the expressions of the multiplier for intermediate values of $L$, we report numerical results here. In order to connect with the literature, we work with the following standard Phillips curve throughout this subsection:
\begin{align*}
\pi_t = \beta\E_t\pi_{t+1} + \kappa\left[\Gamma_c c_t+\Gamma_g g_t - \Gamma_k k_t\right],
\end{align*}
which nests the one we have used so far as a special case. This will also help us illustrate that the results obtained earlier did not hinge on the static Phillips curve assumption. As an organizing principle and to understand better what is driving the aggregate consumption multiplier, we will use the following decomposition:
\begin{align}
\label{eq:decomp_MQQ}
\pazocal{M}_c(L,\theta) \equiv \pazocal{M}_c^{waste}(L,\theta) +  \pazocal{Q}^{deter}_c(L,\theta) +  \pazocal{Q}^{exit}_c(L,\theta), 
\end{align}
where the first component is the public consumption multiplier, the second component picks up the deterministic dynamics before exiting the ELB and the third governs how the dynamics upon exit percolate back to the first period through rational expectations. Note that under Assumption \ref{Ass:eigenvalues}, all the terms on the right hand side do converge to a finite value as $L\to\infty$. In an effort to connect with the numerical results that can be found in the literature, we will also describe what happens under a deterministic exit as in \cite{Guerrieri2014} whenever it is relevant. Note that this is the reason why the component associated with wasteful spending in equation \eqref{eq:decomp_MQQ} now depends on $L$: under stochastic exit, $\pazocal{M}^{waste}_c(L,\theta) = \pazocal{M}^{waste}_c(\theta)$ while under deterministic exit this is not the case.

Let us first study what happens under stable eigenvalues with a stochastic exit.\footnote{To guarantee the stable eigenvalues, we set $\kappa=0.001$.} Combined with the results of subsection \ref{sec:short_lived}, the results in Proposition \ref{prop:c_1pi_1_long_lived} are indicative of how the consumption multiplier varies with $L$. If all the eigenvalues are stable, then for $L=0$ the consumption multiplier is strictly positive. For $L\to\infty$, it will necessarily be lower. As a result, one can expect the output multiplier to be decreasing as a function of $L$. 

\begin{figure}[htp]
\centering
\caption{Decomposition under stable eigenvalues - stochastic exit}
\includegraphics[width=0.8\textwidth]{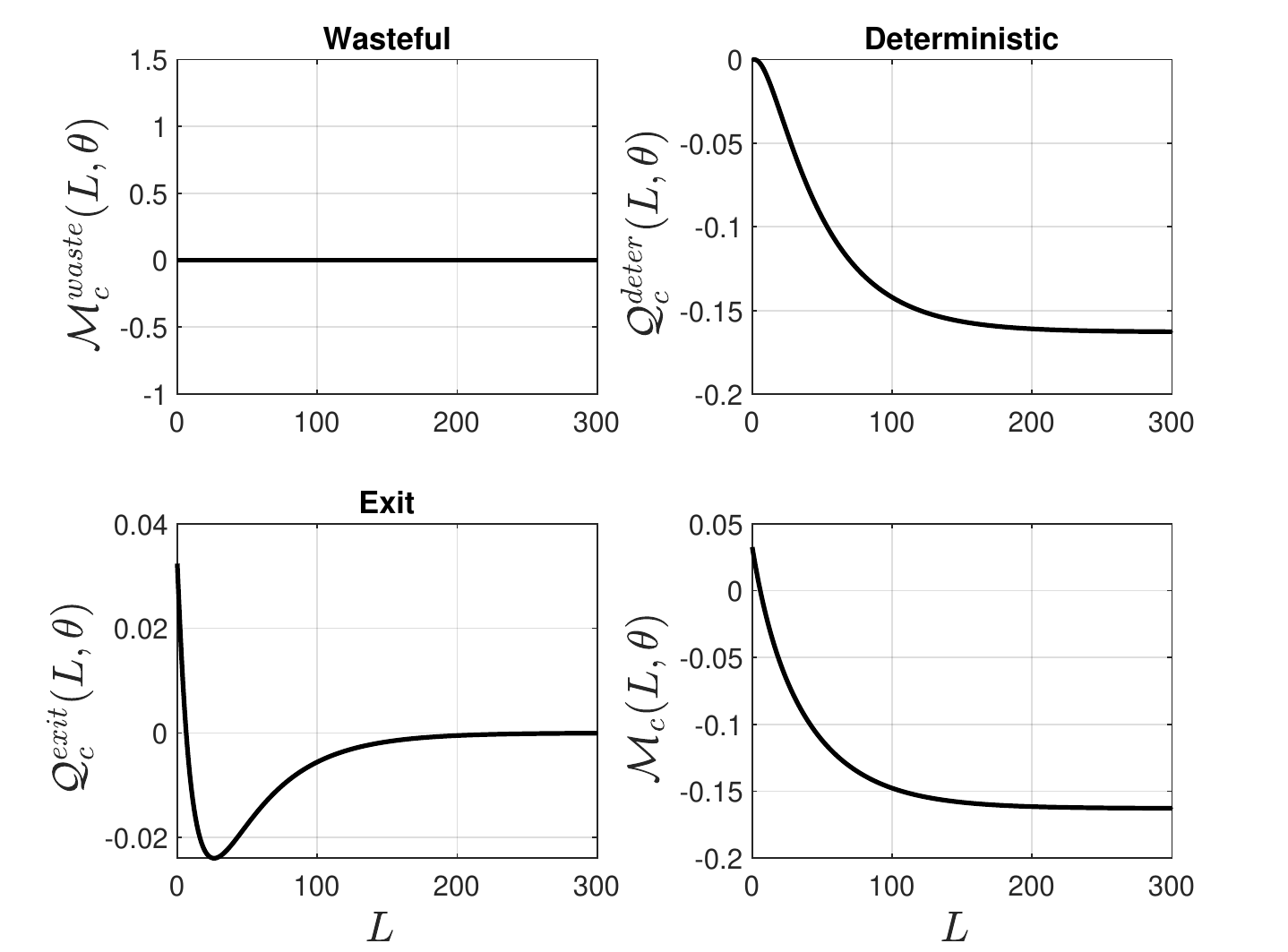}
\begin{minipage}{0.8\textwidth}
\protect\footnotesize Notes: \protect\scriptsize
We use the same calibration as in Figure \ref{fig:AS_AD_NTimes}, except that we now set $\kappa=0.001$.
\end{minipage}
\label{fig:Mc_stable_hybrid}
\end{figure}

First, note that all the terms converge to a finite value as $L$ grows larger. Second, consistent with the results exposed before, the component associated with public consumption is constant. In our case, it is also very close to zero. This is due to the fact that, in our model with a forward-looking Phillips curve the consumption multiplier with respect to public consumption is proportional to $\kappa$, which we have set to a very low value.\footnote{It will not be constant anymore under a deterministic exit, but it will still remain quantitatively negligible regardless of the value of $L$. See Online Appendix figure \ref{fig:Mc_unstable_q_deter}.} Third, notice that in line with the results of Propositions \ref{prop:c_1pi_1_short_lived} and \ref{prop:c_1pi_1_long_lived}, the consumption multiplier is strictly larger than 0 on impact and lower for large values of $L$. One can still use the intuition developed in the previous sections to understand these results. Consistent with the intuition developed in the last subsections, the negative effects come from the dynamics inside the trap (top-right panel of Figure \ref{fig:Mc_stable_hybrid}) where lower inflation is detrimental.

If any of the eigenvalues is unstable however, the public investment multiplier necessarily diverges as a function of $L$. To illustrate that, we plot in Figure \ref{fig:Mc_unstable_q_hybrid} the consumption multiplier as a function of $L$ with the same calibration as in Figure \ref{fig:Mc_stable_hybrid}, but now with $q=0.98$ instead. This makes one of the eigenvalues of $q\mathbf{A}^*$ unstable.\footnote{For our particular calibration, the threshold is given by $\overline{p}= 0.9731$.} More precisely, for this particular calibration we assume stochastic exit as well as $p<\overline{p}<q$. Notice that the consumption multiplier diverges towards $-\infty$: the longer the economy is expected to stay in the liquidity trap, the lower the public investment multiplier.\footnote{In Online Appendix Figure \ref{fig:Mc_unstable_q_deter}, we plot the same experiment for a deterministic exit and show that the same pattern emerges.}

\begin{figure}[htp]
\centering
\caption{Decomposition with one unstable eigenvalue - stochastic exit}
\includegraphics[width=0.8\textwidth]{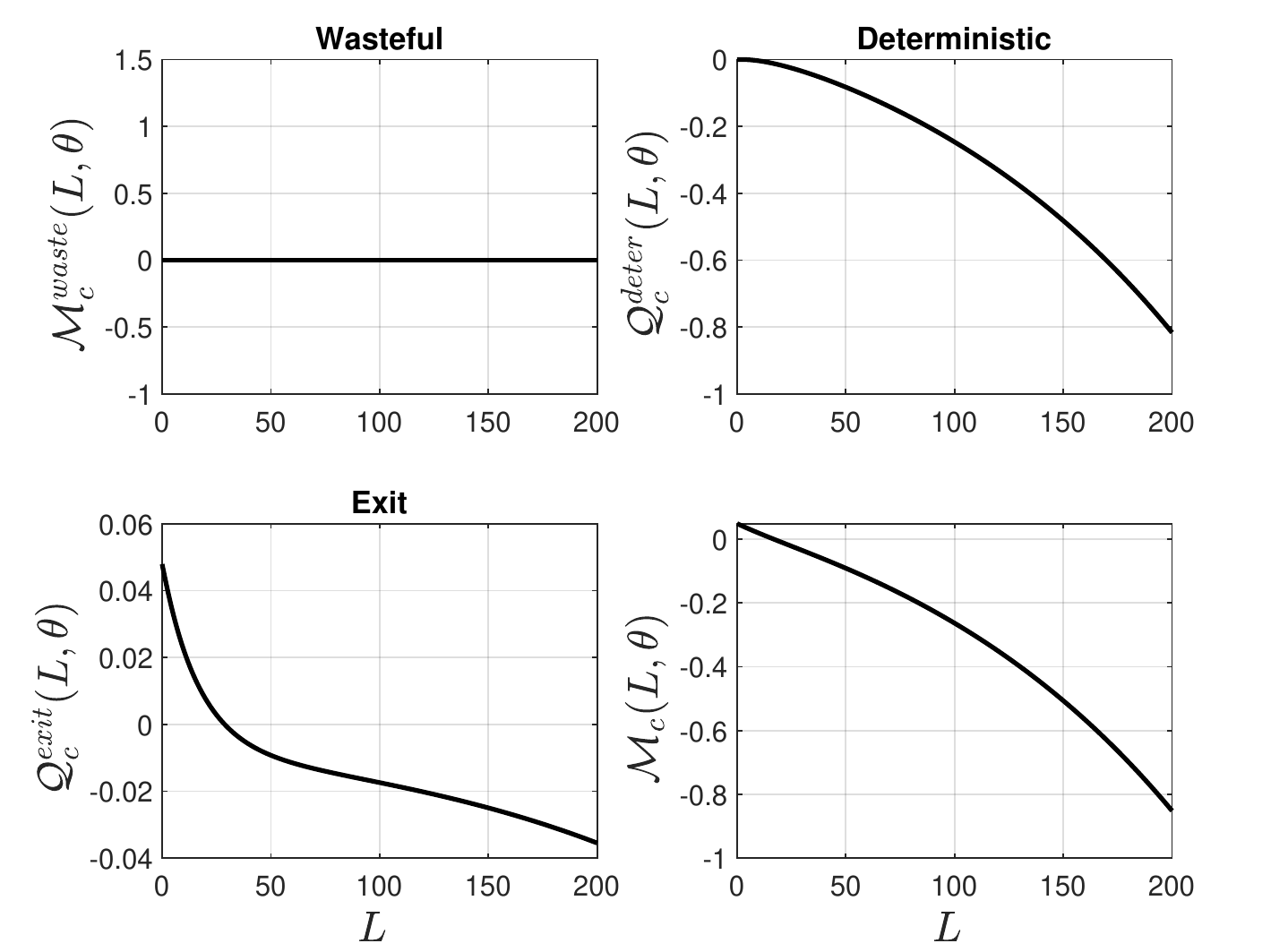}
\begin{minipage}{0.8\textwidth}
\protect\footnotesize Notes: \protect\scriptsize
We use the same calibration as in Figure \ref{fig:Mc_stable_hybrid}, except that we now set $q=0.98$.
\end{minipage}
\label{fig:Mc_unstable_q_hybrid}
\end{figure}

Notice also that the anticipated effects upon exit are expansionary for relatively low values of $L$. This is because for a short exit the aggregate supply effects barely kick in at first. This echoes our results from before. As $L$ gets higher however, the aggregate supply effects from public investment dominate the expected wealth effect quantitatively. This results in anticipated deflation while at the ELB, which is contractionary.  

Lastly, we examine a configuration where both $p$ and $q$ are above the threshold $\overline{p}$. Interestingly, this case encompasses the numerical results provided in \cite{Bouakez2017}. It turns out that in this case, the nature of the exit is crucial. Let us begin with stochastic exit. Notice that now there is a clear tension between what happens during the trap and upon exit. During the trap, there is no uncertainty due to the assumption of perfect foresight: the effects of government investment are solely due to \textit{shifts} in the Euler equation. At the ELB, current consumption does not depend on current inflation and the slope of the Euler equation is flat. Given that aggregate supply effects dominate, the Euler equation shifts down at the ELB. With a positively sloping Phillips curve, this means lower consumption/inflation at the ELB.

\begin{figure}[htp]
\centering
\caption{Decomposition with two unstable eigenvalues - stochastic exit}
\includegraphics[width=0.8\textwidth]{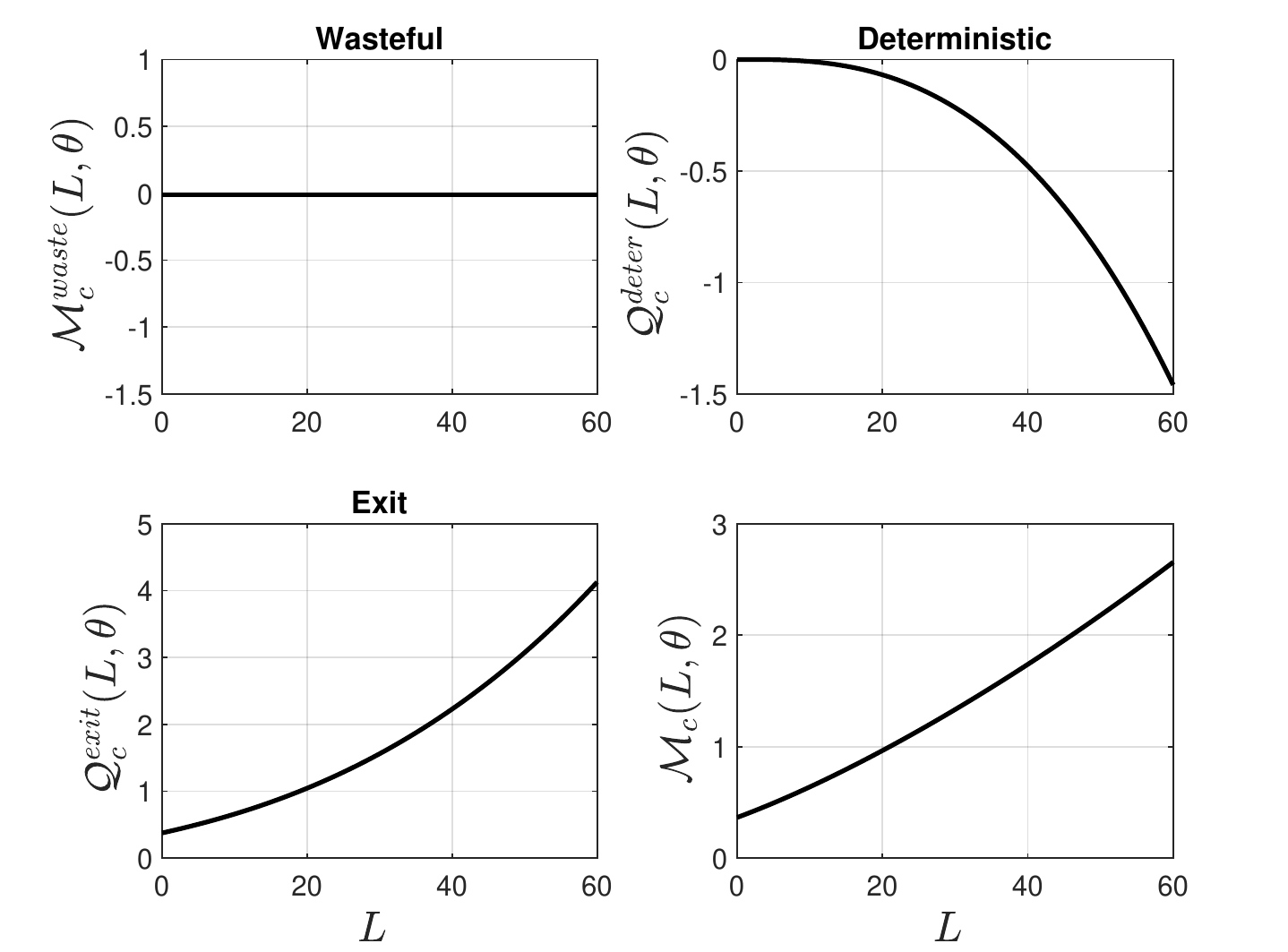}
\begin{minipage}{0.8\textwidth}
\protect\footnotesize Notes: \protect\scriptsize
We use the same calibration as in Figure \ref{fig:Mc_unstable_q_hybrid}, except that we now set $\overline{p}<q<p=0.99$.
\end{minipage}
\label{fig:Mc_unstable_pq_hybrid}
\end{figure}

Upon exit however, there \textit{is} uncertainty and the Euler equation is upward sloping as in \cite{Mertens2014}. As before, this is due to the fact that in this situation, $\E_{L+1}\mathbf{\Pi}_{t+L+1}=p \pi_{L+1} + (1-p)\pi_{L+2}$ whereas under perfect foresight expected inflation is effectively exogenous. The same logic holds for expected consumption. Given that $p>\overline{p}$, we know that the slope of the Euler equation at the ELB is larger in magnitude compared to the Phillips curve. As a result, a downward shift of the Euler equation is now expansionary at the ELB: income effects dominate as in \cite{Bilbiie2022neo} and a positive aggregate supply shock is expansionary. For the same reason, the component associated with public consumption is now negative just as in \cite{Mertens2014}, even if the small magnitude makes it hard to see in Figure \ref{fig:Mc_unstable_pq_hybrid}.

Compared to the other cases we have studied in this subsection, now the nature of the exit will matter in a non trivial way. Accordingly, we now consider a setup with both $p$ and $q$ are above the threshold and deterministic exit. This encompasses the approach taken in \cite{Bouakez2017} who use the algorithm developed in \cite{Guerrieri2014}, except that we do not cover time-to-build. This experiment is reported in Figure \ref{fig:Mc_unstable_pq_deter}. 

\begin{figure}[htp]
\centering
\caption{Decomposition with two unstable eigenvalues - deterministic exit}
\includegraphics[width=0.8\textwidth]{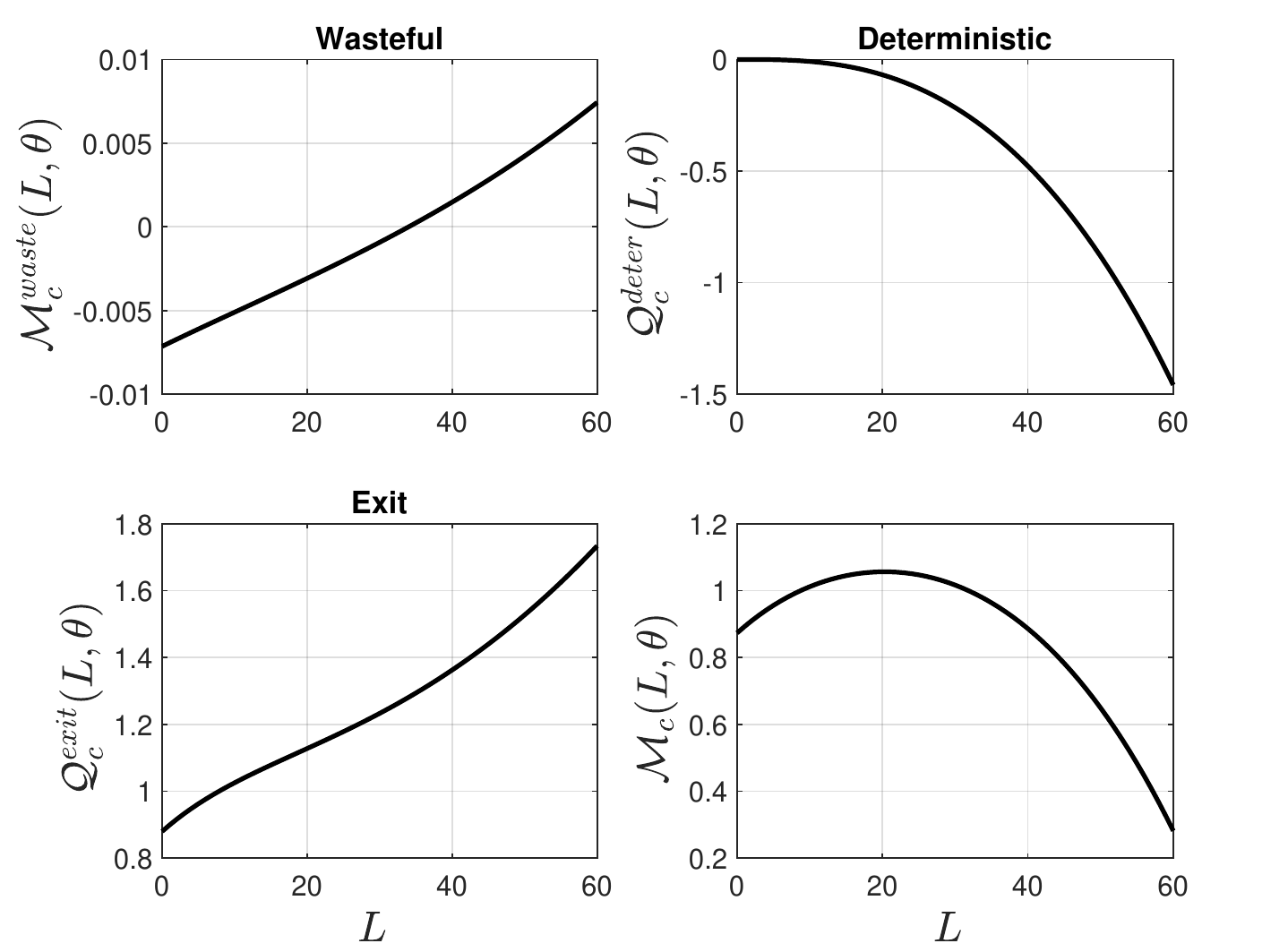}
\label{fig:Mc_unstable_pq_deter}
\end{figure}

In this case, notice that the component associated with public consumption is now increasing in $L$: the longer the trap, the more the inflationary effects of public consumption crowd in private consumption. As with the stochastic exit case, what happens under perfect foresight during the trap is contractionary, while the anticipated effects upon exit are expansionary. Crucially however, now the dynamics upon exit are weighed with matrix $\mathbf{A}$ instead of $\mathbf{A}^*$ for the stochastic exit case. Given that the eigenvalues of $r\mathbf{A}$ for $r=p,q$ are necessarily stable, the dynamics upon exit generate less amplification. In terms of mechanisms, this happens because the Taylor rule is active upon exit and mitigates the aggregate supply effects of public capital. For these reasons, these effects are eventually dominated quantitatively by the ones at the ELB as $L$ grows larger. As a consequence, there is a value of $L$ for which the public investment multiplier is maximal. In our case, this happens for 20 quarters or 5 years. 

This puts the results of \cite{Bouakez2017} in perspective. In their model, the presence of time to build pushes the aggregate supply effects out of the liquidity trap. In our model which corresponds to theirs with a one quarter time to build delay, we see that the mechanisms involved turn out to be more complicated. As $L$ increases, a bigger portion of the supply effects happen at the liquidity trap. Contrary to the intuition developed in \cite{Bouakez2017} however, this actually \textit{increases} the multiplier because of the counteracting effects described in the last paragraph. Eventually however, as $L$ grows larger the intuition developed in \cite{Bouakez2017} prevails. We believe that this highlights the value of the approach that we have developed in this paper: it allows us to peer into the underlying mechanisms with greater detail.  

\subsection{Discussion}

Before concluding, we want to say a few words on a number of related issues: the connection of our results with the existing theoretical and empirical literature as well as the policy relevance of our findings. 

Our results using a standard New Keynesian model point to the following conclusions: $(i)$ in normal times the government investment multiplier is low on impact but may be large in PDV terms and $(ii)$ in a recession the one situation under which the government investment multiplier is relatively low is when the liquidity trap is long as a result of a large but not too persistent shock. In all other cases, we find a relatively large government investment multiplier (at least in PDV terms) where private consumption is crowded in. Our finding of overall high public investment multipliers at the lower bound is consistent with the results of \cite{Bouakez2017} as well as \cite{Tervala2022building}.

Regarding the connection to the empirical literature, our results broadly agree with the low impact multipliers outside the lower bound reported in \cite{Abiad2016}, \cite{Boehm2019} as well as the large multipliers under constant interest rates reported in \cite{Acconcia2014}, \cite{Boehm2019}
and 
\cite{Morita2020empirical}. 

While we have presented our results in a standard New Keynesian model, these are more general and apply to an extended version that may include heterogeneous agents or deviates from the assumption of full information rational expectations \textemdash as long as said model does not feature any endogenous state variables. Under any of these extensions, the conditions for stability of the eigenvalues will most likely change. For example, a model with heterogeneous agents may make the explosive dynamics more prevalent. In that case, our results show that choosing between the deterministic algorithm developed in \cite{Guerrieri2014} or the stochastic one that we introduce in this paper is not innocuous in terms of policy implications. We would advise the interested policymaker to try both for reasonable values of $L$ to make sure that the conclusions do not depend on a particular choice of solution method.  

\section{Conclusion}
\label{sec:conclusion}

In this paper, we have proposed a simple framework to think about the aggregate effects of public investment. We have kept the model deliberately simple in order to clearly expose the driving forces behind public investment. The model produces results that are in line with the existing theoretical and empirical literature, but clarifies that under some circumstances the public investment multiplier may explode as the duration of the liquidity trap increases. To check whether that is the case or not, we have provided clear conditions that can readily be checked with other versions of the standard New Keynesian model.

With that being said, the results that we have developed in this paper hold more generally for any tractable New Keynesian model, as long as it does not feature endogenous persistence. As a result, this framework is well suited to be applied to the rich literature on tractable heterogeneous agent versions of the standard New Keynesian models that have been developed recently: \cite{Werning2015incomplete}, \cite{Acharya2020understanding}, \cite{Bilbiie2019a,Bilbiie2020new}, \cite{Broer2020new}, \cite{Holm2020monetary}, \cite{Ravn2021macroeconomic}, \cite{Acharya2023optimal} and \cite{Pfauti2022behavioral}. I study this issue in \cite{Roulleau2021public}.

Allowing for endogenous persistence brings in a whole new set of challenges that we have not addressed in this paper. This however has the potential to bring in new and interesting insights for a much larger class of DSGE models compared to the one studied in the present paper. Accordingly, we are following this avenue in current research.  

\bibliographystyle{apalike2}
\bibliography{refs}
\pagebreak

\begin{center}
\Large
\textbf{Online Appendix}
\normalsize    
\end{center}

\appendix

\section{Derivations for the full model}
\label{app:full_model}

The household maximization program can be written with the following Lagrangian:
\begin{align*}
&\mathcal{L}_0^{H}\equiv \E_0\sum_{t=0}^\infty \beta^t\bigg\{\left[\log(C_{t})-\chi\frac{N_{t}^{1+\eta}}{1+\eta}\right]\\
&-\Lambda_t\left[C_t+\frac{B_t}{P_t}-W_tN_t+\mathcal{T}_t-\mathcal{D}_t-\xi_{t-1}\frac{1+R_{t-1}}{P_t}B_{t-1}\right]\bigg\},
\end{align*}
where all the variables are the usual suspects. The first order condition with respect to consumption is given by:
\begin{align}
\nonumber
\frac{\p \mathcal{L}_0^{H}}{\p C_t} &= 0\\
\nonumber
\Leftrightarrow \beta^t \left\{\frac{1}{C_t}-\Lambda_t\right\}&=0\\
\Leftrightarrow \Lambda_t &=  \frac{1}{C_t}.
\label{eq:FOC_C}
\end{align}
The first order condition with respect to hours worked is given by:
\begin{align}
\nonumber
\frac{\p \mathcal{L}_0^{H}}{\p N_t} &= 0\\
\nonumber
\Leftrightarrow \beta^t \left\{-\chi N_t^\eta + \Lambda_t W_t\right\}&=0\\
\Leftrightarrow W_t &=  \chi C_tN_t^\eta,
\label{eq:FOC_N}
\end{align}
where we have used equation \eqref{eq:FOC_C} to substitute for $\Lambda_t$. Finally, the first order condition with respect to government bonds is given by:
\begin{align}
\nonumber
\frac{\p \mathcal{L}_0^{H}}{\p B_{t}} &= 0\\
\nonumber
\Leftrightarrow -\beta^t \frac{\Lambda_t}{P_t} + \beta^{t+1}\E_t\left\{\xi_t\frac{1+R_t}{P_{t+1}}\Lambda_{t+1}\right\} &=0\\
\Leftrightarrow \Lambda_t &= \beta\E_t\left\{ \xi_t\frac{1+R_t}{1+\Pi_{t+1}}\Lambda_{t+1}\right\}.
\label{eq:FOC_B}
\end{align}
Now turning to the supply side, we focus on a symmetric equilibrium where all monopolistically competitive firms use the following production function:
\begin{align}
\label{eq:prodfunc_no_g}
Y_t = K_t^{\epsilon_g} N_t,
\end{align}
where $K_t$ is a stock of public capital which evolves according to the following law of motion:
\begin{align*}
K_t = (1-\delta)K_{t-1} + G_t,    
\end{align*}
where $G_t$ denotes government investment. All firms indexed by $z$ produce a differentiated good and face a Rotemberg-style quadratic cost to adjust their price and seek to maximize expected discounted profits:
\be
\mathbbm{E}_{t}\sum_{s=0}^{\infty }\beta ^{s}%
\frac{C_t}{C_{t+s}}\mathcal{D}_{t+s}(z) ,
\label{Firm's Objective}
\ee
where profits at time $t$ are given by:
\begin{align*}
\mathcal{D}_t(z) &= (1+\tau)Y^D_t(z)-W_tN_t(z)-\frac{\psi }{2}\left( \frac{P_{t}(z)}{%
P_{t-1}}-1\right) ^{2}Y_{t}\label{eq:profits_Gen_Eq}\\
&= (1+\tau)\left(\frac{P_t(z)}{P_t}\right)^{-\nu}Y_t-MC_t\left(\frac{P_t(z)}{P_t}\right)^{-\nu}Y_t-\frac{\psi }{2}\left( \frac{P_{t}(z)}{%
P_{t-1}}-1\right) ^{2}Y_{t},
\end{align*}
where $\nu>1$ is the elasticity of substitution across goods and, as in \cite{Bilbiie2019a}, we assume that the quadratic cost is with respect to the past \textit{aggregate} price. As is standard, we will focus on a symmetric equilibrium where all firms $z$ will make identical choices. In turn, marginal cost are given by 
\begin{align}
MC_t = \frac{W_t N_t}{Y_t}    =\frac{W_t}{K_t^{\epsilon_g}}
\end{align}
The optimal pricing decision gives rise to the following non linear and static Phillips curve:
\begin{align}
-\nu+\nu MC_t-\psi d(\Pi_t)=0,
\label{eq:NKPC}    
\end{align}
where we have assumed the optimal constant subsidy of $\tau=1/(\nu-1)$. Assuming that the government has access to lump-sum taxes, combining firm's profits with the government and household budget constraints yields the resource constraint:
\begin{align}
\label{eq:ge_no_g}
\Delta_t Y_t = C_t+G_t\quad\text{with}\quad \Delta_t\equiv 1-\frac{\psi}{2}\Pi_t^2.  
\end{align}
To close the model, we assume that the Central Bank sets the nominal rate $R_t$ according to a rule that will be specified directly in log-linear terms. 

\subsection{The full model - steady state}
It will be useful to define the following steady-state ratios
\begin{align}
s_c = \frac{C}{Y}\quad \Rightarrow \quad \frac{G}{Y}=1-s_c.
\end{align}
Using the public capital accumulation equation, we get:
\begin{align*}
K &= (1-\delta)K +  G\\
\Rightarrow\quad G &= \delta K\\
\Rightarrow\quad \frac{K}{Y} &= \frac{1-s_c}{\delta}.
\end{align*}
We can write steady state output as
\begin{align*}
Y = \frac{C}{s_c}.    
\end{align*}
Using that to substitute in the steady state version of the intra-temporal labor supply choice, we obtain:
\begin{align*}
\chi CN^\eta &= \frac{Y}{N}  \\
\Rightarrow\quad \chi CN^\eta &= \frac{C}{s_c N}\\
\Rightarrow\quad \chi N^{1+\eta} &= \frac{1}{s_c}\\
\Rightarrow\quad N &= \left[\frac{1}{s_c\chi}\right]^{\frac{1}{1+\eta}},
\end{align*}
where we calibrate $\chi$ so that $N=1/3$ at steady state. Using again the steady state share of consumption, we get:
\begin{align*}
C&=s_cY\\
&=s_cNK^{\epsilon_g}\\
&=s_cN\left(\frac{1-s_c}{\delta}Y\right)^{\epsilon_g}\\
&=s_cN\left(\frac{1-s_c}{\delta}\right)^{\epsilon_g}\left(\frac{C}{s_c}\right)^{\epsilon_g}\\
\Leftrightarrow\ C^{1-\epsilon_g}&=s_c^{1-\epsilon_g}N\left(\frac{1-s_c}{\delta}\right)^{\epsilon_g}\\
\Leftrightarrow\ C&=s_c\left[N\left(\frac{1-s_c}{\delta}\right)^{\epsilon_g}\right]^{\frac{1}{1-\epsilon_g}}.
\end{align*}
Now that we have $C$, we can directly compute $Y$ and thus both $K$ and $G$. Independent of the real allocation, we can compute the return on government bonds using the bonds Euler equation:
\begin{align*}
R = \frac{1}{\beta}-1    
\end{align*}

\subsection{The full model - linear approximation}

The linear approximation of equation \eqref{eq:FOC_C} is given by
\begin{align}
\lambda_t = -c_t,    
\end{align}
where lowercase letters denote percent deviations from steady state. Moving on to equation \eqref{eq:FOC_N}, we get:
\begin{align}
w_t &= c_t+\eta n_t.    
\end{align}
Moving on to the bond Euler equation, we get the following expression:
\begin{align}
c_t = \E_tc_{t+1}-(r_t-\E_t\pi_{t+1}+\xi_t).    
\end{align}
The Phillips curve can be expressed as
\begin{align}
\pi_t=\kappa mc_t,   
\end{align}
where $\kappa=\psi/\nu$. Labor demand implies that real marginal costs can be written as:
\begin{align}
\nonumber
mc_t &= w_t + n_t - y_t\\
\nonumber
&= w_t -\epsilon_g k_t\\
&= c_t +\eta n_t -\epsilon_g k_t
\end{align}
The production function is now given by
\begin{align}
y_t = \epsilon_g k_t + n_t 
\end{align}
Letting $g_t$ denote $(G_t-G)/Y$, the resource constraint is given by
\begin{align}
y_t = s_c c_t + g_t.    
\end{align}
Using these last two equations to rewrite the real marginal cost, we obtain:
\begin{align}
\nonumber
mc_t &= c_t +\eta (y_t-\epsilon_g k_t) -\epsilon_g k_t   \\
\nonumber
&= c_t + \eta (s_c c_t + g_t) - (1+\eta)\epsilon_g k_t \\
\nonumber
&= (1+\eta s_c)c_t +\eta g_t - (1+\eta)\epsilon_g k_t\\
&\equiv \Gamma_c c_t+\Gamma_g g_t-\Gamma_k k_t,
\end{align}
where we have defined
\[
\Gamma_c = 1+\eta s_c\quad\&\quad\Gamma_g = \eta\quad\&\quad \Gamma_k = (1+\eta)\epsilon_g.
\]
Finally, the law of motion for public capital can be approximated as;
\begin{align*}
k_t &= (1-\delta)k_{t-1}+\frac{G}{K}\frac{G_t-G}{G}\\
    &= (1-\delta)k_{t-1}+\delta\frac{G_t-G}{G}\\
    &= (1-\delta)k_{t-1}+\delta\frac{Y}{G}\frac{G_t-G}{Y}\\
    &= (1-\delta)k_{t-1}+\frac{\delta}{1-s_c}g_t\\
    &\equiv  (1-\delta)k_{t-1}+\tilde{\delta}g_t
\end{align*}

Regrouping equations and imposing the ELB, we end up with the following system of equations:
\begin{align}
	c_t	& = \mathbb{E}_tc_{t+1} - [r_t  - \mathbb{E}_t\pi_{t+1} + \xi_t] \label{eq:full_EE} \\
	\pi_t	& =\kappa \left[ (1+\eta s_c)c_t+\Gamma_g g_t-\Gamma_k k_t\right] \label{eq:full_NKPC} \\
	r_t	& = \max[\log(\beta); \phi_\pi\pi_t], \label{eq:full_TR}\\
	k_t	& = (1-\delta)k_{t-1}+\tilde{\delta}g_t \label{eq:full_K_LoM}\\
    y_t	& = s_cc_t + g_t \label{eq:full_RC} 
\end{align}
Notice that we can use equations \eqref{eq:full_K_LoM}-\eqref{eq:full_TR} to uniquely solve for $c_t,\pi_t,R_t$ for a given path of $g_t$ and $k_t$. We can then use equation \eqref{eq:full_RC} to back out the log deviation of aggregate output. Assuming that the ELB does not bind, we can re-arrange this system as follows:
\begin{align*}
c_t +\phi_\pi\pi_t   &= \E_t c_{t+1} + \E_t\pi_{t+1} -\xi_t\\
-\kappa\Gamma_c c_t+ \pi_t &= \beta\E_t\pi_{t+1}-\kappa \Gamma_k k_t +\kappa\Gamma_g g_t
\end{align*}
As a result, casting this in the matrix framework in the main text implies:
\begin{align*}
\mathbf{A}_0
&=
\begin{bmatrix}
1 & \phi_\pi\\
-\kappa\Gamma_c & 1
\end{bmatrix},
\mathbf{A}_1
=
\begin{bmatrix}
1 & 1\\
0 & \beta
\end{bmatrix}
, 
B_0=
\begin{bmatrix}
0\\
-\kappa\Gamma_k
\end{bmatrix}
,\\
C_{0,1}
&=
\begin{bmatrix}
0 \\
\kappa\Gamma_g
\end{bmatrix},
C_{0,2}
=
\begin{bmatrix}
-1 \\
0
\end{bmatrix},
\end{align*}
and where $d=(1-q)/(1-s_c)$. Now turning to the case where the ZLB is binding, we have the following system of equations:
\begin{align*}
c_t   &= \E_t c_{t+1} + \E_t\pi_{t+1} -\xi_t -\log(\beta)\\
-\kappa\Gamma_c c_t+ \pi_t &= \beta\E_t\pi_{t+1}-\kappa\Gamma_k k_t +\kappa\Gamma_g g_t
\end{align*}
As a result, casting this in the matrix framework in the main text implies:
\begin{align*}
\mathbf{A}^*_0
&=
\begin{bmatrix}
1 & 0\\
-\kappa\Gamma_c & 1
\end{bmatrix},
\mathbf{A}^*_1
=
\begin{bmatrix}
1 & 1\\
0 & \beta
\end{bmatrix}
, 
B^*_0=
\begin{bmatrix}
0\\
-\kappa\Gamma_k
\end{bmatrix}
,\\
C^*_{0,1}
&=
\begin{bmatrix}
0 \\
\kappa\Gamma_g
\end{bmatrix},
C^*_{0,2}
=
\begin{bmatrix}
-1 \\
0
\end{bmatrix},
E_0^*
=
\begin{bmatrix}
-\log(\beta) \\
0
\end{bmatrix},
\end{align*}
and where we still have $d=(1-q)/(1-s_c)$.
\pagebreak

\section{Definition of Markov states for Definition \ref{ddef:MChains_inertial}}
\label{sec:app_MChains_inertial}

For the preference and government spending shocks, the Markov states are given by
\begin{align*}
z_{\ell} &=
\begin{cases}
z_{1} \quad &\text{for}\ \ell=1\\
pz_{\ell-1}\quad &\text{for}\ \ell=2,\dots,L+1\\
0\quad &\text{for}\ \ell\in\left\{L+2,L+3\right\},
\end{cases}    
\end{align*}
for $z\in\left\{g,\xi\right\}$. This effectively ensures that both exogenous shocks follow an $AR(1)$ with persistence $p$ in expectations. The Markov states for the inertial exogenous variable are given by the vector $\left[k_1,\dots,k_{L+2},0\right]$. We need to make sure that these replicate an $ARMA(2,1)$ in expectations. Given the one period time to build delay, we know that $k_t=0$ on impact, while
\[
k_{t+n} = \frac{q^{n}-p^{n}}{q-p}k_{t+1}
\]
for $n\geq 1$ after an initial shock to $g_t$. From equation \eqref{eq:ll_K_LoM}, this implies that the first state is given by $k_2=\Tilde{\delta}g_1$. To solve for $k_3,\dots,k_{L+1}$, let us remember that
\[
\E_t \mathbf{K}_{t+\ell} = k_{\ell+1}
\]
for $\ell=1,\dots,L$ given the assumption of perfect foresight and the 0 diagonal elements. As a result, we can set:
\[
k_{\ell+1} = \frac{q^{\ell}-p^{\ell}}{q-p}k_2
\]
for the Markov states indexed by $\ell=0,\dots,L$. In turn, this implies that $k_{L+2}$ must be such that:
\begin{align*}
\E_t \mathbf{K}_{t+L+1} &= pk_{L+1}+(1-p)k_{L+2} = \frac{q^{L+1}-p^{L+1}}{q-p}k_2\\
\Rightarrow (1-p)k_{L+2} &= \frac{q^{L+1}-p^{L+1}}{q-p}k_2 -pk_{L+1}\\
\Rightarrow k_{L+2} &=\frac{1}{1-p}\left(\frac{q^{L+1}-p^{L+1}}{q-p}-p\frac{q^{L}-p^{L}}{q-p}\right)k_2 \\
\Rightarrow k_{L+2} &= \frac{k_2}{1-p}q^{L}\\
&= \frac{\Tilde{\delta}g_1}{1-p}q^{L}
\end{align*}
for $L\geq 0$. Note that if $L=0$ so that the ELB only binds for the first period in expectations, then this boils down to the Markov chain transition matrix developed in \cite{Roulleau2023analyzing} with $k_2=\Tilde{\delta}g_1/(1-p)$.

\pagebreak

\section{Proof of Proposition \ref{prop:match_IRF}}
\label{sec:proof_match_IRF}
By construction and given the assumption of perfect foresight for the first $L$ periods, the Markov chain matches both the forward and the backward equations for a given guess of $L$. What remains to be proven is that the Markov chain solves the forward and backward equations for time periods $L+1,L+2,\dots$. Let us begin with the backward-looking variable $k_t$ first. From Definition \ref{ddef:MChains_inertial}, we have
\begin{align*}
k_{L+1} &= \frac{q^{L}-p^{L}}{q-p}k_2\\
k_{L+2} &= \frac{k_2}{1-p}q^L\\
k_{L+3} &= 0.
\end{align*}
Using the Chapman-Kolmogorov theorem, we can write for $m\geq 0$:
\begin{align*}
\E_t\mathbf{K}_{t+L+m} &= u\mathcal{P}^{L+m}\mathbf{K}  \\
&= p^mk_{L+1} + (1-p)\frac{q^m-p^m}{q-p}k_{L+2}\\
&= p^m\frac{q^{L}-p^{L}}{q-p}k_2 + (1-p)\frac{q^m-p^m}{q-p}\frac{k_2}{1-p}q^L\\
&= p^m\frac{q^{L}-p^{L}}{q-p}k_2 + \frac{q^m-p^m}{q-p}k_2q^L\\
&= \frac{k_2}{q-p}\left[p^m(q^{L}-p^{L}) + (q^m-p^m)q^L\right]\\
&= \frac{q^{L+m}-p^{L+m}}{q-p}k_2,
\end{align*}
which matches the deterministic $AR(2)$ formulation for $k_t$ in Definition \ref{ddef:MChains_inertial}. Let us now turn to the system of forward equations. As before, these trivially hold for the first $L$ periods. To show that it also holds for time periods $L+1$ onward, we will proceed by induction as in \cite{Roulleau2023analyzing}. To accommodate for the fact that the dynamics in normal time start at time $L+2$, notice that the dynamics from time $L+1$ onward can be represented by a Markov chain with the transition probability for the case $L=0$, but with a vector of Markov states  $\left[Y_{L+1},\ Y_{L+2}, 0_{2,1}\right]$ for the forward-looking variables and likewise for the other variables. In some sense, the perfect foresight assumption implies that the chain somewhat forgets the dynamics before time $L+1$: the first $L$ states are transient and will never be visited again as the chain converges to its absorbing state. With this in mind and using the fact that the ELB is exactly binding at time $L+1$ we can write that
\begin{align}
\nonumber
\mathbf{A}_0\E_t\mathbf{Y}_{t+L+1} &= \mathbf{A}_1\E_t \mathbf{Y}_{t+L+2}+B_0\E_t \mathbf{K}_{t+L+1}+\mathbf{C}_0\E_t\mathbf{Z}_{t+L+1}  \\
\Leftrightarrow\quad \mathbf{A}_0Y_{L+1} &= p\mathbf{A}_1Y_{L+1} + (1-p)\mathbf{A}_1Y_{L+2} + B_0k_{L+1}+\mathbf{C}_0Z_{L+1}
\label{eq:EtYtlp1}\\
&= p\mathbf{A}^*_1Y_{L+1} + (1-p)\mathbf{A}^*_1Y_{L+2} + B_0^*k_{L+1}+\mathbf{C}_0^*Z_{L+1} +E^*,
\nonumber
\end{align}
which is the definition of $Y_{L+1}$ given in the main text and where we have regrouped the Markov states for both exogenous shocks in the vector $Z$. Likewise, using the fact that the ELB is not binding anymore in period $L+2$ notice that 
\begin{align*}
\mathbf{A}_0\E_t\mathbf{Y}_{t+L+2} &= \mathbf{A}_1\E_t \mathbf{Y}_{t+L+3}+B_0\E_t \mathbf{K}_{t+L+2}+\mathbf{C}_0\E_t\mathbf{Z}_{t+L+2} \\
\Leftrightarrow\quad \mathbf{A}_0(pY_{L+1}+(1-p)Y_{L+2}) &= \mathbf{A}_1(p^2Y_{L+1} + (1-p)(q+p)Y_{L+2})\\
&+ B_0(pk_{L+1}+(1-p)k_{L+2})\\
&+ \mathbf{C}_0(pZ_{L+1}+(1-p)Z_{L+2}).
\end{align*}
Using equation \eqref{eq:EtYtlp1} to substitute for $p\mathbf{A}_0Y_{L+1}$ on the left hand side, we get:
\begin{align*}
(1-p)\mathbf{A}_0Y_{L+2} &= q(1-p)\mathbf{A}_1Y_{L+2} + (1-p)B_0k_{L+2} +   (1-p)\mathbf{C}_0Z_{L+2}.
\end{align*}
Dividing both sides by $(1-p)$, we get the expression for $Y_{L+2}$ given in the main text. As a result, we have shown that the expressions for $Y_{L+1}$ and $Y_{L+2}$ are such that the forward equation holds for time $L+1$ and $L+2$. Accordingly, let us now assume that the forward equation holds for generic time periods $L+m,L+m+1$ for $m\geq 1$. We will now show that if that is the case, then it also holds for time period $L+m+2$. As in \cite{Roulleau2023analyzing}, we will use the fact that 
\begin{align*}
\E_t\mathbf{Y}_{t+L+m+2} = (p+q)\E_t\mathbf{Y}_{t+L+m+1}-pq\E_t\mathbf{Y}_{t+L+m}    
\end{align*}
for $m\geq 0$ given the structure of the Markov chain and likewise for Markov chains $\mathbf{K}_t$ and $\mathbf{Z}_t$. Using this property, we can write that the following linear recurrence holds:
\begin{align*}
\mathbf{A}_0\E_t\mathbf{Y}_{t+L+m+2} &= (p+q)\mathbf{A}_0\E_t\mathbf{Y}_{t+L+m+1}-pq\E_t\mathbf{A}_0\mathbf{Y}_{t+L+m}\\
&= (p+q)\left[\mathbf{A}_1\E_t\mathbf{Y}_{t+L+m+2}+B\E_t \mathbf{K}_{t+L+m+1}+\mathbf{C}\E_t\mathbf{Z}_{t+L+m+1}\right]\\
&-pq\left[\mathbf{A}_1\E_t\mathbf{Y}_{t+L+m+1}+B\E_t \mathbf{K}_{t+L+m}+\mathbf{C}\E_t\mathbf{Z}_{t+L+m}\right]\\
&= \mathbf{A}_1\E_t\mathbf{Y}_{t+L+m+3} + B\E_t \mathbf{K}_{t+L+m+2}+\mathbf{C}\E_t\mathbf{Z}_{t+L+m+2}
\end{align*}
where we have used the induction hypothesis on the second line and the linear recurrence on the last line. We now need to make sure that the ELB is binding exactly at time $L+1$. To do so, the vector $Y_{L+1}$ has to be such that:
\begin{align*}
Y_{L+1} &= p\mathbf{A}^*Y_{L+1}+(1-p)\mathbf{A}^*Y_{L+2}+B^*k_{L+1}+C^*_{1}g_{L+1}+\text{t.i.p}\\
&= p\mathbf{A}^*Y_{L+1}+(1-p)\mathbf{A}^*\left(I-q\mathbf{A}\right)^{-1}Bk_{L+2}+B^*k_{L+1}+C^*_{1}g_{L+1}+ \text{t.i.p}\\
&= p\mathbf{A}^*Y_{L+1}+q^{L}\mathbf{A}^*\left(I-q\mathbf{A}\right)^{-1}Bk_{2}+\frac{q^{L}-p^{L}}{q-p}B^*k_{2}+p^LC^*_{1}g_1+ \text{t.i.p}\\
&= p\mathbf{A}^*Y_{L+1}+q^{L}\mathbf{A}^*\left(I-q\mathbf{A}\right)^{-1}B\tilde{\delta}g_1+\frac{q^{L}-p^{L}}{q-p}B^*\tilde{\delta}g_1+p^LC^*_{1}g_1+ \text{t.i.p}\\
&= \left(I-p\mathbf{A}^*\right)^{-1}\left[q^{L}\mathbf{A}^*\left(I-q\mathbf{A}\right)^{-1}B\tilde{\delta}+\frac{q^{L}-p^{L}}{q-p}B^*\tilde{\delta}+p^LC^*_{1}\right]g_1+\text{t.i.p}\\
&\equiv \Theta_{L,1} g_1+p^L\Theta_{2} \xi_1+\Tilde{E}^*
\end{align*}
We can now use this equation to compute the value of $\xi_1$ such that the ELB is exactly binding at time $t+L+1$. For this, we assume that the policy shock $g_1\sim 0$. Given that inflation is the second element in the vector of forward looking variables, we need to have $\xi_1$ such that:
\begin{align}
p^L\Theta_{2}(2,1)\xi_1 + \Tilde{E}^*(2,1) = \frac{\log(\beta)}{\phi_\pi},
\label{eq:condition_z21}
\end{align}
which can easily be solved for $\xi_1$ for a given $L$. This only proves that the constraint is exactly binding at time $t+L+1$, but not necessarily before that. As a result, we need to make sure that the ELB is indeed binding for all periods $t,t+1,\dots, t+L$. To do so, we will first work under the assumption that it is indeed binding for these periods and solve for the Markov states for inflation and consumption under this assumption. Then we will use these states to show that the ELB is indeed binding. The following proposition exposes this result.
\begin{prop}
\label{prop:zlb_binds}

Assume that $\xi_1$ satisfies equation \eqref{eq:condition_z21} for a given $L$. It follows that 
\begin{align*}
\E_t\mathbf{\Pi}_{t+\ell}< \frac{\log(\beta)}{\phi_\pi}\quad \Rightarrow \quad -\log(\beta) + \E_t\mathbf{R}_{t+\ell} = 0 
\end{align*}
so that the ELB is binding for all $\ell=1,\dots,L$.
\end{prop}
\begin{proof}
To prove that the ELB is indeed binding, we will use the fact that the policy shock $g_1\sim 0$. This will allow us to write:
\begin{align*}
Y_{\ell} &= \mathbf{A}^*Y_{\ell+1}+B^*k_\ell+p^{\ell-1}C^*_{2}\xi_1 + E^{*},
\end{align*}
for $\ell=L,\dots,1$ and where we have used the recursive expression for the Markov states of the exogenous shock. By definition, the shock $\xi_1$ is set up so that
\begin{align}
\label{eq:induction_pi}
\E_t\mathbf{\Pi}_{t+L} = \pi_{L+1} = \frac{\log(\beta)}{\phi_\pi}
\end{align}
It follows that in this case we also have :
\begin{align}
\label{eq:induction_c}
\E_t\mathbf{C}_{t+L}=c_{L+1} = \frac{1}{\kappa\Gamma_c}\pi_{L+1} = \frac{1}{\kappa\Gamma_c} \frac{\log(\beta)}{\phi_\pi}
\end{align}
given the assumption of stochastic exit from the ELB. We will use induction to show that 
\[
\pi_{L+1-k} \leq \frac{\log(\beta)}{\phi_\pi}
\]
for $k=0,\dots,L$. Notice first that it must be the case that $\xi_1>0$ for it to be possible that the shadow interest rate comes back to the steady state from below.

With this in mind, the case $k=0$ is true by construction. Let us now assume that the hypothesis is true for $k=0,\dots,L-1$. In that case, we can write 
\begin{align*}
c_{L+1-k-1} &= c_{L+1-k} + \pi_{L+1-k} -\log(\beta) - p^{L+1-k-2}\xi_1    \\
            &=\left(1 + \kappa\Gamma_c \right)c_{L+1-k}-\log(\beta) - p^{L+1-k-2}\xi_1
\end{align*}
Using the induction hypothesis, we can write:
\begin{align}
c_{L+1-k-1} &\leq \left(1 + \kappa\Gamma_c \right)\frac{\log(\beta)}{ \kappa\Gamma_c\phi_\pi}-\log(\beta) - p^{L+1-k-2}\xi_1
\label{eq:induction_step}
\end{align}
Using the Euler equation for $c_{L+1}$, we can write:
\begin{align*}
c_{L+1} &= pc_{L+1} + (1-p)c_{L+2}-p^{L}\xi_1 -\log(\beta)\\
\Rightarrow\quad c_{L+1} &= c_{L+2} - \frac{1}{1-p}[p^{L}\xi_1 +\log(\beta)]
\end{align*}
By construction, we have that 
\begin{align*}
c_{L+1}\leq  \frac{\log(\beta)}{\kappa\Gamma_c\phi_\pi} < c_{L+2}.
\end{align*}
As a result, we can write that;
\begin{align*}
- \frac{1}{1-p}[p^{L}\xi_1 +\log(\beta)] &=  c_{L+1}-c_{L+2}   \\
\Rightarrow\quad \frac{1}{1-p}[p^{L}\xi_1 +\log(\beta)] &< 0 \\
\Rightarrow\quad p^L\xi_1 > -\log(\beta)
\end{align*}
Given that $p\in(0,1)$, this implies that
\begin{align*}
\xi_1 > p\xi_1 > \dots > p^{L-1}\xi_1 > p^{L}\xi_1 > -\log(\beta)     
\end{align*}
Now returning to the induction step, we have
\begin{align*}
c_{L+1-k-1} &\leq \left(1 + \kappa\Gamma_c \right)\frac{\log(\beta)}{ \kappa\Gamma_c\phi_\pi}-\log(\beta) - p^{L+1-k-2}\xi_1 \\
\Rightarrow\quad &\leq \left(1 + \kappa\Gamma_c \right)\frac{\log(\beta)}{ \kappa\Gamma_c\phi_\pi}\\
\Rightarrow\quad &\leq \frac{\log(\beta)}{ \kappa\Gamma_c\phi_\pi}.
\end{align*}
Now multiplying both sides by $\kappa\Gamma_c>0$, we find that
\begin{align*}
\pi_{L+1-k-1} \leq \frac{\log(\beta)}{\phi_\pi}.
\end{align*}
As a result, assuming that the constraint is binding for $k$ implies that it is also binding for $k=1$. That allows us to conclude that the lower bound binds for all the periods before the exit. We have shown that the induction hypothesis is true. It follows that 
\begin{align*}
\E_t\mathbf{\Pi}_{t+\ell}< \frac{\log(\beta)}{\phi_\pi}\quad \Rightarrow \quad -\log(\beta) + \E_t\mathbf{R}_{t+\ell} = 0    
\end{align*}
for $\ell=1,\dots,L$ and with equality for $\ell=L+1$, where the left hand side of the second equation represents the impulse response of the nominal interest rate in level. 
\end{proof}
Now that we have shown that the ELB is indeed binding for time periods $1,\dots,L+1$ we can work backwards to solve for the vectors $Y_1,\dots,Y_L$ as follows:
Given the perfect foresight assumption, we can write:
\begin{align*}
Y_{\ell}&= \mathbf{A}^*Y_{\ell+1}+B^*k_\ell+C^*_{1}z_{1,\ell}+\text{t.i.p}    \\
         &= \mathbf{A}^*Y_{\ell+1}+\left(\frac{q^{\ell-1}-p^{\ell-1}}{q-p}B^*d+p^{\ell-1}C^*_{1}\right)g_1+\text{t.i.p}    
\end{align*}
for all $\ell=L-1,\dots,1$. We can get the following general expression:
\begin{align*}
Y_{L-\ell} &= (\mathbf{A}^*)^{\ell+1}Y_{L+1}+p^{L-\ell-1}\left(\sum_{i=0}^{\ell}(p\mathbf{A}^*)^i\right)C^*_{1}g_1+\text{t.i.p}\\
&+\frac{1}{q-p}\left[q^{L-\ell-1}\sum_{i=0}^{\ell}(q\mathbf{A}^*)^i-p^{L-\ell-1}\sum_{i=0}^{\ell}(p\mathbf{A}^*)^i\right]B^*dg_1.  
\end{align*}
In particular, to get the impact effect we set $\ell=L-1$ and obtain:
\begin{align*}
Y_1 &= (\mathbf{A}^*)^LY_{L+1}+\left(\sum_{i=0}^{L-1}(p\mathbf{A}^*)^i\right)C^*_{1}g_1\\
&+\frac{1}{q-p}\left[\sum_{i=0}^{L-1}(q\mathbf{A}^*)^i-\sum_{i=0}^{L-1}(p\mathbf{A}^*)^i\right]B^*dg_1  +\text{t.i.p}  \\
&= (\mathbf{A}^*)^LY_{L+1}+\left(I-(p\mathbf{A}^*)^L\right)\left(I-p\mathbf{A}^*\right)^{-1} C^*_{1}g_1+\text{t.i.p}\\
&+ \frac{1}{q-p}\left[\left(I-(q\mathbf{A}^*)^L\right)\left(I-q\mathbf{A}^*\right)^{-1}-\left(I-(p\mathbf{A}^*)^L\right)\left(I-p\mathbf{A}^*\right)^{-1}\right]B^*dg_1.
\end{align*}
Note that we can compute a similar equilibrium in which the ELB does binds in state $L$ but does not bind in state $L+1$. In this case we have $T=L$ and $Y_{L+1}$ can be solved for using equation \eqref{eq:EtYtlp1}. In that case, $\xi_1$ can be calibrated as being in a range such that
\begin{align*}
p^{L-1}\Theta_{2}(2,1)\xi_1 + \Tilde{E}^*(2,1) \leq  \frac{\log(\beta)}{\phi_\pi} < p^{L-1}\Theta_{2}(2,1)\xi_1.
\end{align*}
This effectively replicates the algorithm developed in \cite{Guerrieri2014}.

\pagebreak
\section{Proof of Proposition \ref{prop:pi_1_NTimes}}
\label{sec:proof_pi1_NTimes}

We reproduce the short run Euler equation and Phillips curve for convenience:
\begin{align}
\label{eq:c_1_Ntimes_app}
c_1 &= pc_1 + (1-p)c_2 -\left[\phi_\pi \pi_1 - p\pi_1 - (1-p)\pi_2\right]\\
\pi_1 &= \kappa\left[\Gamma_c c_1 + \Gamma_g g\right].
\label{eq:pi_1_Ntimes_app}
\end{align}
Collecting the $c_1$ terms and re-arranging terms in equation \eqref{eq:c_1_Ntimes_app}, we get:
\begin{align*}
c_1 = c_2+\pi_2-\frac{\phi_\pi-p}{1-p}\pi_1.    
\end{align*}
Using the expressions for medium run inflation and consumption as a function of $k=\tilde{\delta}g$, we get:
\begin{align*}
c_1 &= \frac{\kappa(\phi_\pi-1)\Gamma_k}{1-q+\kappa\Gamma_c(\phi_\pi-q)}\frac{\tilde{\delta}}{1-p}g-\frac{\phi_\pi-p}{1-p}\pi_1\\
    &=  \frac{\kappa(\phi_\pi-1)(1+\Gamma_g)}{\Delta(\phi_\pi,q,\cdot)}\frac{\tilde{\delta}}{1-p}g\epsilon_g-\frac{\phi_\pi-p}{1-p}\pi_1
\end{align*}
where we have used the expression for $\Gamma_k=(1+\Gamma_g)\epsilon_g$ as well as defined $\Delta(\phi_\pi,q,\cdot)= 1-q+\kappa\Gamma_c(\phi_\pi-q)$ on the second line. Using now equation \eqref{eq:pi_1_Ntimes_app} to substitute for $\pi_1$ and re-arrange, we get:
\begin{align*}
\frac{\partial c_1}{\partial g} = \kappa\frac{\frac{(\phi_\pi-1)(1+\Gamma_g)}{\Delta(\phi_\pi,q,\cdot)}\tilde{\delta}\epsilon_g - (\phi_\pi-p)\Gamma_g}{\Delta(\phi_\pi,p,\cdot)}    .
\end{align*}
Notice that the denominator on the right hand side is strictly positive for $p\in(0,1)$. As a result, the sign of the fraction depends only on its denominator. For given structural parameters, the denominator is an affine function of $\epsilon_g$. For the special case of $\epsilon_g=0$, the numerator is strictly negative. In addition, the numerator is strictly increasing in $\epsilon_g$ and goes to $+\infty$ as $\epsilon_g\to +\infty$. As a result, we know that there exists a threshold value $\epsilon_g^I$ such that 
\begin{align*}
\frac{\partial c_1}{\partial g}>0 \quad\text{if}\quad \epsilon_g>  \epsilon_g^I.   
\end{align*}
It is straightforward to see that the threshold value is given by
\begin{align*}
\epsilon_g^I = \frac{\phi_\pi-p}{\phi_\pi-1} \frac{\Gamma_g}{1+\Gamma_g} \frac{\Delta(\phi_\pi,q,\cdot)}{\tilde{\delta}}>0.
\end{align*}
This completes the proof of parts 1 and 2. To prove part 3, start with the Phillips curve and obtain:
\begin{align*}
\frac{\partial \pi_1}{\partial g} = \kappa\left[\Gamma_c\frac{\partial c_1}{\partial g} + \Gamma_g\right].    
\end{align*}
As a result, the effect is positive only if the term in brackets sum to a positive number. Using the short run multiplier effect on consumption, we get:
\begin{align*}
\frac{\partial \pi_1}{\partial g}>0\quad & \Leftrightarrow\quad  \Gamma_c\frac{\partial c_1}{\partial g} + \Gamma_g>0\\
 & \Leftrightarrow \kappa\Gamma_c \frac{\frac{(\phi_\pi-1)(1+\Gamma_g)}{\Delta(\phi_\pi,q,\cdot)}\tilde{\delta}\epsilon_g - (\phi_\pi-p)\Gamma_g}{\Delta(\phi_\pi,p,\cdot)} + \Gamma_g>0\\
 & \Leftrightarrow  \frac{\kappa\Gamma_c\frac{(\phi_\pi-1)(1+\Gamma_g)}{\Delta(\phi_\pi,q,\cdot)}\tilde{\delta}\epsilon_g -\kappa\Gamma_c (\phi_\pi-p)\Gamma_g + \Delta(\phi_\pi,p,\cdot)\Gamma_g}{\Delta(\phi_\pi,p,\cdot)} >0\\
 & \Leftrightarrow  \frac{\kappa\Gamma_c\frac{(\phi_\pi-1)(1+\Gamma_g)}{\Delta(\phi_\pi,q,\cdot)}\tilde{\delta}\epsilon_g  + (\Delta(\phi_\pi,p,\cdot)-\kappa\Gamma_c (\phi_\pi-p))\Gamma_g}{\Delta(\phi_\pi,p,\cdot)} >0.
 \end{align*}
Given the definition of $\Delta(\phi_\pi,p,\cdot)$, notice that
\begin{align*}
\Delta(\phi_\pi,p,\cdot)-\kappa\Gamma_c (\phi_\pi-p) &=  1-q+\kappa\Gamma_c(\phi_\pi-p)   -\kappa\Gamma_c (\phi_\pi-p)\\
&= 1-q>0.
\end{align*}
Given the fact that the first term on the denominator is strictly positive given parameter restrictions, it follows that the multiplier
effect on inflation in the short run is necessarily positive. 

\pagebreak

\section{Proof of Proposition \ref{prop:Mc_NTimes}}
\label{sec:proof_Mc_NTimes}

From Appendix \ref{sec:proof_pi1_NTimes}
, we know that
\begin{align*}
\frac{\partial c_1}{\partial g} = \frac{\Theta\tilde{\delta}\epsilon_g-\kappa(\phi_\pi-p)\Gamma_g}{\det(I-p\mathbf{A})}.
\end{align*}
As a result, we can write the cumulative PDV multiplier as
\begin{align*}
\mathfrak{M}_c = \frac{\Theta\tilde{\delta}\epsilon_g-\kappa(\phi_\pi-p)\Gamma_g}{\det(I-p\mathbf{A})}+\frac{\beta}{1-\beta q}\Theta_{c,k}\tilde{\delta}\epsilon_g
\end{align*}
This cumulative PDV multiplier is strictly positive if and only if
\begin{align*}
\frac{\Theta\tilde{\delta}\epsilon_g}{\det(I-p\mathbf{A})} + \frac{\beta}{1-\beta q}\Theta_{c,k}\tilde{\delta}\epsilon_g &> \frac{\kappa(\phi_\pi-p)\Gamma_g}{\det(I-p\mathbf{A})} \\
\Leftrightarrow\quad \frac{\Theta\tilde{\delta}\epsilon_g}{\det(I-p\mathbf{A})} + \frac{\beta}{1-\beta q}\Theta_{c,k}\tilde{\delta}\epsilon_g &> \frac{\kappa(\phi_\pi-p)\Gamma_g}{\det(I-p\mathbf{A})} \\
\Leftrightarrow\quad \epsilon_g + \det(I-p\mathbf{A})\frac{\beta}{1-\beta q}\frac{\Theta_{c,k}}{\Theta}\epsilon_g &> \epsilon_g^I\\
\Leftrightarrow\quad \epsilon_g  &> \frac{\epsilon_g^I}{1+\det(I-p\mathbf{A})\frac{\beta}{1-\beta q}\frac{\Theta_{c,k}}{\Theta}}\\
\Leftrightarrow\quad \epsilon_g  &> \frac{\epsilon_g^I}{1+\det(I-p\mathbf{A})\frac{\beta}{1-\beta q}\frac{\phi_\pi-q}{\phi_\pi-1}},
\end{align*}
where we have used the fact that $\Theta_{c,k}/\Theta=(\phi_\pi-q)/\phi_\pi-1$ on the last line. This completes the proof.

\pagebreak
\section{Proof of Proposition \ref{prop:eigs}}
\label{sec:proof_eigs}

At the ELB, we have:
\begin{align*}
(\mathbf{A}^*_0)^{-1}\mathbf{A}^*_1 \equiv \mathbf{A}^* = 
\begin{bmatrix}
1 & 1\\
\kappa\Gamma_c & \kappa\Gamma_c
\end{bmatrix}.
\end{align*}
If $\lambda$ is an eigenvalue of $\mathbf{A}$, then it solves the following polynomial equation:
\begin{align*}
\det\left(\mathbf{A}-\lambda I\right) = 0 \quad &\Leftrightarrow\quad \lambda^2 - (1+\kappa\Gamma_c) \lambda = 0\\
&\Leftrightarrow\quad \lambda\left[\lambda-(1+\kappa\Gamma_c)\right]=0.
\end{align*}
As a result, this matrix has two distinct eigenvalues given by
\begin{align*}
\lambda_1 = 0\quad\&\quad \lambda_2 = 1+\kappa\Gamma_c.    
\end{align*}
The first eigenvalue is then stable. Further, if $\lambda_1$ is an eigenvalue of $\mathbf{A}$ then $p\lambda_1$ is an eigenvalue of $p\mathbf{A}$. As a result, this eigenvalue of $p\mathbf{A}$ is necessarily stable as well.

Given that $\kappa\Gamma_c>0$, it follows that $\lambda_2>1$. Then, the second eigenvalue of $p\mathbf{A}$ is unstable if and only if 
\begin{align*}
\lambda_2>\frac{1}{p}\quad &\Leftrightarrow \quad 1+\kappa\Gamma_c>\frac{1}{p}\\
&\Leftrightarrow p > \overline{p} \equiv \frac{1}{1+\kappa\Gamma_c} \in (0,1),
\end{align*}
where we have used the fact that $\lambda_2>1$ to guarantee that $\overline{p}\in (0,1)$. We now move on to the second part of the Proposition. Using the expression for $\mathbf{A}$, we can write:
\begin{align*}
I-p\mathbf{A} = 
\begin{bmatrix}
1-p & -p\\
-p\kappa\Gamma_c & 1-p\kappa\Gamma_c.
\end{bmatrix}    
\end{align*}
Therefore, we can compute
\begin{align*}
\det\left(I-p\mathbf{A}\right) &= (1-p)(1-p\kappa\Gamma_c) -p^2\kappa\Gamma_c\\
&= 1-p(1+\kappa\Gamma_c)\\
&\equiv f(p).
\end{align*}
Note that $f$ is strictly decreasing for $p\in(0,1)$ given that $f^{'}(p)=-(1+\kappa\Gamma_c)$. Further, note that $f(0)=1>0$ while $f(1)=-\kappa\Gamma_c<0$. As a result, $f$ there exists a value of $p=\tilde{p}$ such that 
\begin{itemize}
    \item $f(p)>0$  for $p<\tilde{p}$
    \item $f(p)=0$ for $p=\tilde{p}$
    \item $f(p)<0$  for $p>\tilde{p}$.
\end{itemize}
Given that $f(p)$ is an affine function, we can write $\tilde{p}$ as:
\begin{align*}
\tilde{p} = \frac{1}{1+\kappa\Gamma_c} = \overline{p}. 
\end{align*}
This completes the proof. 

\pagebreak 
\section{Proof of Proposition \ref{prop:c_1pi_1_short_lived}}
\label{sec:proof_c_1pi_1_short_lived}

We reproduce the short run Euler equation and Phillips curve for convenience:
\begin{align}
\label{eq:c_1_Short_lived_app}
c_1 &= pc_1 + (1-p)c_2 + p\pi_1 + (1-p)\pi_2 - \xi\\
\pi_1 &= \kappa\left[\Gamma_c c_1 + \Gamma_g g\right].
\label{eq:pi_1_Short_lived_app}
\end{align}
Collecting the $c_1$ terms and re-arranging terms in equation \eqref{eq:c_1_Ntimes_app}, we get:
\begin{align*}
c_1 = c_2+\pi_2+\frac{p}{1-p}\pi_1.    
\end{align*}
Using the expressions for medium run inflation and consumption as a function of $k=\tilde{\delta}g/(1-p)$, we get:
\begin{align*}
c_1 &= \Theta\frac{\tilde{\delta}}{1-p}\epsilon_g g+\frac{p}{1-p}\pi_1.
\end{align*}
Using now equation \eqref{eq:pi_1_Short_lived_app} to substitute for $\pi_1$ and re-arrange, we get:
\begin{align*}
\frac{\partial c_1}{\partial g} = \frac{\Theta\tilde{\delta}  \epsilon_g+p\kappa\Gamma_g}{\det(I-p\mathbf{A}^*)}    .
\end{align*}
Now using equation \eqref{eq:pi_1_Short_lived_app}, we get:
\begin{align*}
\frac{\partial \pi_1}{\partial g} &= \kappa\Gamma_c\left[\Theta\frac{\Tilde{\delta}}{1-p}\epsilon_g+\frac{p}{1-p}\frac{\partial \pi_1}{\partial g}\right]+\kappa\Gamma_g     \\
\Leftrightarrow\quad (1-p)\frac{\partial \pi_1}{\partial g} &= \kappa\Gamma_c\Theta\Tilde{\delta}\epsilon_g+p\kappa\Gamma_c\frac{\partial \pi_1}{\partial g} + (1-p)p\kappa\Gamma_g\\
\frac{\partial \pi_1}{\partial g} &= \kappa\frac{\Theta\Tilde{\delta}\epsilon_g\Gamma_c+(1-p)\Gamma_g}{\det(I-p\mathbf{A}^*)},
\end{align*}
where we have used the fact that $\det(I-p\mathbf{A}^*)=1-p(1+\kappa\Gamma_c)$ on the last line. This completes the proof. 

\pagebreak
\section{Proof of Proposition \ref{prop:PDV_long_trap}}
\label{sec:proof_PDV_long_trap}

From Appendix \ref{sec:proof_c_1pi_1_short_lived}
, we know that
\begin{align}
\frac{\partial c_1}{\partial g} = \frac{\Theta\tilde{\delta}\epsilon_g+p\kappa\Gamma_g}{\det(I-p\mathbf{A}^*)}.
\label{eq:short_lived_det_neg}
\end{align}
As a result, we can write the cumulative PDV multiplier as
\begin{align*}
\mathfrak{M}_c = \frac{\Theta\tilde{\delta}\epsilon_g+p\kappa\Gamma_g}{\det(I-p\mathbf{A}^*)}+\frac{\beta}{1-\beta q}\Theta_{c,k}\tilde{\delta}\epsilon_g
\end{align*}
In this case, given that $\det(I-p\mathbf{A}^*)$ is negative, the short run effect of public capital is negative while the medium run effect is positive. The second one dominates if and only if
\begin{align}
\frac{\beta}{1-\beta q}\Theta_{c,k}   > -\frac{\Theta}{\det(I-p\mathbf{A}^*)}>0.
\label{eq:condition_short_lived_det_neg}
\end{align}
Using the fact that
$\frac{\Theta_{c,k}}{\Theta} = \frac{\phi_\pi-q}{\phi_\pi-1}$,
condition \eqref{eq:condition_short_lived_det_neg} boils down to:
\begin{align*}
\beta\frac{\phi_\pi-q}{1-\beta q}>-\frac{\phi_\pi-1}{\det(I-p\mathbf{A}^*)}>0    .
\end{align*}
If this condition holds, then one can view $\mathfrak{M}_c$ as an affine function $\mathcal{F}(\epsilon_g)$ that is such that $\mathcal{F}(0)<0$ as well as $\mathcal{F}'(\epsilon_g)>0$. As a result, there exists a threshold $\epsilon^{M,z}_g$ such that if $\epsilon_g>\epsilon^{M,z}_g$ then $\mathfrak{M}_c>0$. From equation \eqref{eq:short_lived_det_neg}, it is clear that if condition \eqref{eq:condition_short_lived_det_neg} then 
\[
\mathfrak{M}_c<\frac{\partial c_1}{\partial g}<0.
\]
This completes the proof.

\pagebreak
\section{Proof of Proposition \ref{prop:long-lived_trap}}
\label{sec:proof_long-lived_trap}

Using the results from Appendix \ref{sec:proof_match_IRF}, we can compute the vector of impact states as:
\begin{align*}
Y_1 &= (\mathbf{A}^*)^LY_{L+1}+\left(I-(p\mathbf{A}^*)^L\right)\left(I-p\mathbf{A}^*\right)^{-1} C^*_{1}g+\text{t.i.p}\\
&+ \frac{1}{q-p}\left[\left(I-(q\mathbf{A}^*)^L\right)\left(I-q\mathbf{A}^*\right)^{-1}-\left(I-(p\mathbf{A}^*)^L\right)\left(I-p\mathbf{A}^*\right)^{-1}\right]B^*\Tilde{\delta}g.
\end{align*}
Using the expression for $Y_{L+1}$ derived before, we can rewrite the first term on the right hand side as follows:
\begin{align*}
(\mathbf{A}^*)^LY_{L+1} &=  (\mathbf{A}^*)^L\left(I-p\mathbf{A}^*\right)^{-1}\left[q^{L}\mathbf{A}^*\left(I-q\mathbf{A}\right)^{-1}B\Tilde{\delta}+\frac{q^{L}-p^{L}}{q-p}B^*d+p^LC^*_{1}\right]g_1 +\text{t.i.p}  \\
&= (q\mathbf{A}^*)^L\left(I-p\mathbf{A}^*\right)^{-1}\mathbf{A}^*\left(I-q\mathbf{A}\right)^{-1}B\Tilde{\delta}g_1\\
&+ \frac{1}{q-p} (q\mathbf{A}^*)^L\left(I-p\mathbf{A}^*\right)^{-1}B^*\Tilde{\delta}g_1\\
&-\frac{1}{q-p} (p \mathbf{A}^*)^L\left(I-p\mathbf{A}^*\right)^{-1}B^*\Tilde{\delta}g_1\\
&+(p \mathbf{A}^*)^L\left(I-p\mathbf{A}^*\right)^{-1}C^*_1g_1+\text{t.i.p}
\end{align*}
Therefore, we can rewrite the vector of impact states as
\begin{align*}
Y_1 &=  (p \mathbf{A}^*)^L\left(I-p\mathbf{A}^*\right)^{-1}C^*_1g +  \left(I-(p\mathbf{A}^*)^L\right)\left(I-p\mathbf{A}^*\right)^{-1} C^*_{1}g+\pazocal{Q}(L,\theta)g+\text{t.i.p}\\
&= \left(I-p\mathbf{A}^*\right)^{-1}C^*_1g+\pazocal{Q}(L,\theta)g+\text{t.i.p},
\end{align*}
where $\pazocal{Q}(L,\theta)$ collects all the terms involving $g$ through public capital. Under assumptions \ref{Ass:eigenvalues}-\ref{Ass:big_shock}, notice that both $(p\mathbf{A}^*)^L$ and $(q\mathbf{A}^*)^L$ converge to a matrix of zeros as $L\to\infty$. The same goes for the t.i.p terms which are finite and multiplied by $(p\mathbf{A}^*)^L$. Again using assumptions \ref{Ass:eigenvalues}-\ref{Ass:big_shock}, notice that 
\begin{align*}
\lim_{L\to\infty}  \left(I-(q\mathbf{A}^*)^L\right)\left(I-q\mathbf{A}^*\right)^{-1} &= \left(I-q\mathbf{A}^*\right)^{-1}  \\
\lim_{L\to\infty}  \left(I-(p\mathbf{A}^*)^L\right)\left(I-p\mathbf{A}^*\right)^{-1} &= \left(I-p\mathbf{A}^*\right)^{-1}  
\end{align*}
Using these two properties, we can write that
\begin{align*}
\lim_{L\to\infty} Y_1 &= \left(I-p\mathbf{A}^*\right)^{-1} C^*_{1}g+\text{t.i.p}\\
&+ \frac{1}{q-p}\left[\left(I-q\mathbf{A}^*\right)^{-1}-\left(I-p\mathbf{A}^*\right)^{-1}\right]B^*\Tilde{\delta}g.
\end{align*}
We now need to relate this expression to the one that we have derived in the Introduction. To do so, notice that
\begin{align*}
(q-p)    \left(I-p\mathbf{A}^*\right)^{-1}\mathbf{A}^*\left(I-q\mathbf{A}^*\right)^{-1} &= \left(I-q\mathbf{A}^*\right)^{-1}-\left(I-p\mathbf{A}^*\right)^{-1}\\
\Leftrightarrow\quad (q-p)\mathbf{A}^*\left(I-q\mathbf{A}^*\right)^{-1} &= \left(I-p\mathbf{A}^*\right)\left(I-q\mathbf{A}^*\right)^{-1}-I\\
\Leftrightarrow\quad (q-p)\mathbf{A}^* &= (I-p\mathbf{A}^*)-(I-q\mathbf{A}^*)\\
&= q\mathbf{A}^*-p\mathbf{A}^*,
\end{align*}
where we have left-multiplied both sides by $(I-p\mathbf{A}^*)$ on the second line and right-multiplied by $(I-q\mathbf{A}^*)$ on the third line. This guarantees that the equality on the first line holds. This proves that
\begin{align*}
\lim_{L\to\infty} \pazocal{Q}(L,\theta) &=     \frac{1}{q-p}\left[\left(I-q\mathbf{A}^*\right)^{-1}-\left(I-p\mathbf{A}^*\right)^{-1}\right]B^*\Tilde{\delta}\\
&= \left(I-p\mathbf{A}^*\right)^{-1}\mathbf{A}^*\left(I-q\mathbf{A}^*\right)^{-1}B^*\Tilde{\delta},
\end{align*}
which completes the proof.

\pagebreak
\section{Proof of Proposition \ref{prop:c_1pi_1_long_lived}}
\label{sec:proof_c_1pi_1_long_lived}

We reproduce the short run Euler equation and Phillips curve for convenience:
\begin{align}
\label{eq:c_1_Short_lived_app}
c_1 &= pc_1 + (1-p)c_2 + p\pi_1 + (1-p)\pi_2 - \xi\\
\pi_1 &= \kappa\left[\Gamma_c c_1 + \Gamma_g g\right].
\label{eq:pi_1_Short_lived_app}
\end{align}
Collecting the $c_1$ terms and re-arranging terms in equation \eqref{eq:c_1_Ntimes_app}, we get:
\begin{align*}
c_1 = c_2+\pi_2+\frac{p}{1-p}\pi_1.    
\end{align*}
Using the expressions for medium run inflation and consumption as a function of $k=\tilde{\delta}g/(1-p)$, we get:
\begin{align*}
c_1 &= \Theta^z\frac{\tilde{\delta}}{1-p}\epsilon_g g+\frac{p}{1-p}\pi_1.
\end{align*}
Using now equation \eqref{eq:pi_1_Short_lived_app} to substitute for $\pi_1$ and re-arrange, we get:
\begin{align*}
\frac{\partial c_1}{\partial g} = \frac{\Theta^z\tilde{\delta}  \epsilon_g+p\kappa\Gamma_g}{\det(I-p\mathbf{A}^*)}    .
\end{align*}
Now using equation \eqref{eq:pi_1_Short_lived_app}, we get:
\begin{align*}
\frac{\partial \pi_1}{\partial g} &= \kappa\Gamma_c\left[\Theta^z\frac{\Tilde{\delta}}{1-p}\epsilon_g+\frac{p}{1-p}\frac{\partial \pi_1}{\partial g}\right]+\kappa\Gamma_g     \\
\Leftrightarrow\quad (1-p)\frac{\partial \pi_1}{\partial g} &= \kappa\Gamma_c\Theta^z\Tilde{\delta}\epsilon_g+p\kappa\Gamma_c\frac{\partial \pi_1}{\partial g} + (1-p)p\kappa\Gamma_g\\
\frac{\partial \pi_1}{\partial g} &= \kappa\frac{\Theta^z\Tilde{\delta}\epsilon_g\Gamma_c+(1-p)\Gamma_g}{\det(I-p\mathbf{A}^*)},
\end{align*}
where we have used the fact that $\det(I-p\mathbf{A}^*)=1-p(1+\kappa\Gamma_c)$ on the last line. This completes the proof. 

\pagebreak
\section{Additional figures}
\label{sec:additional_figs}

\begin{figure}[htp]
\centering
\caption{Output multiplier as a function of $L$ - one unstable eigenvalue, deterministic exit}
\includegraphics[width=0.8\textwidth]{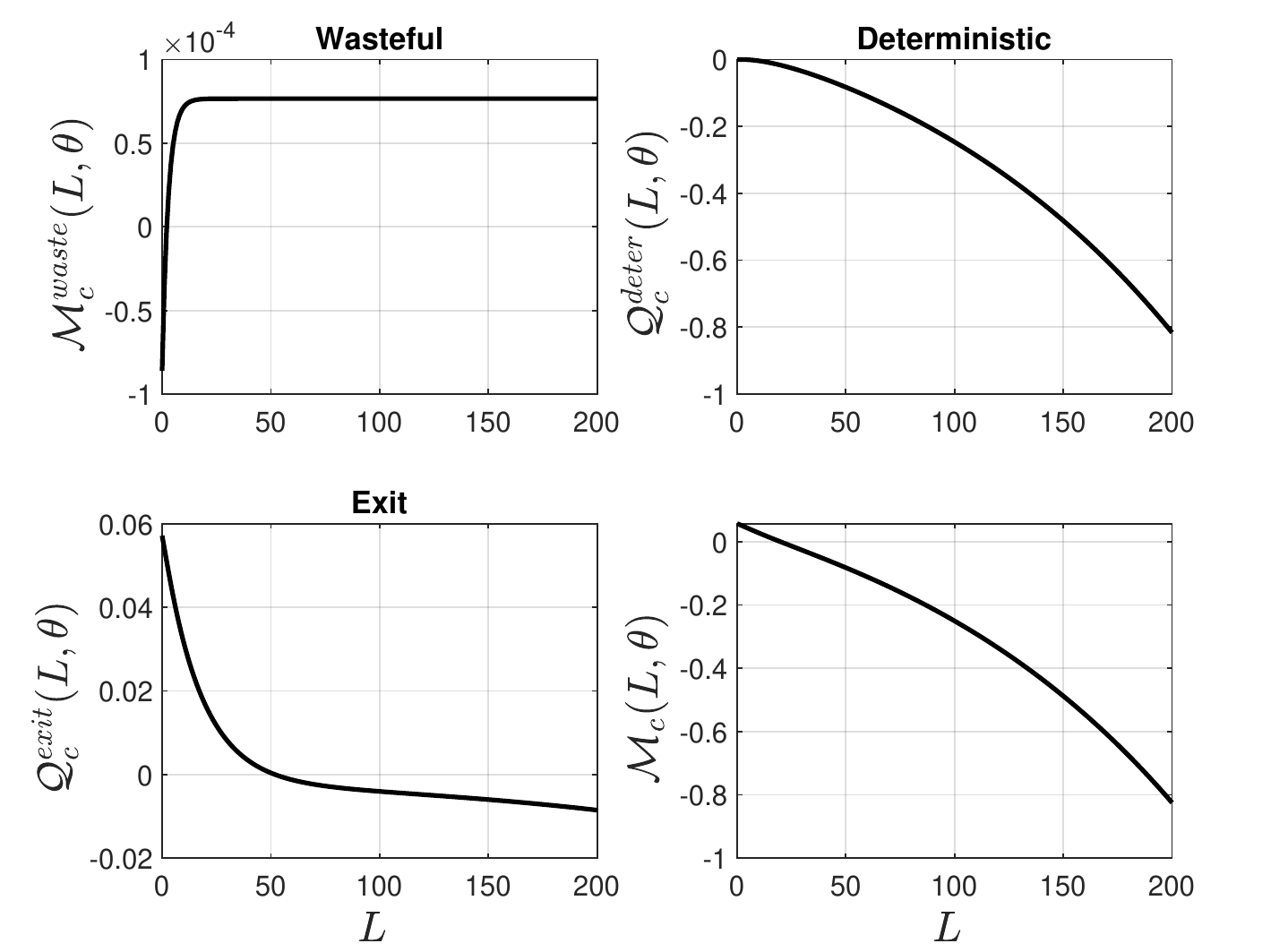}
\begin{minipage}{0.8\textwidth}
\protect\footnotesize Notes: \protect\scriptsize
We use the same calibration as in Figure \ref{fig:Mc_unstable_q_hybrid}, except that we now assume a deterministic exit.
\end{minipage}
\label{fig:Mc_unstable_q_deter}
\end{figure}

\end{document}